\documentclass[journal]{IEEEtran}
\usepackage[english]{babel}
\usepackage{amsmath,amssymb}
\usepackage{amsthm,empheq}
\usepackage{physics}
\usepackage{amsfonts}
\usepackage{mathtools}
\usepackage{graphicx}
\usepackage{caption}
\usepackage{color}
\usepackage{optidef}
\usepackage{hyperref}
\usepackage{enumitem}
\usepackage{subcaption}
\usepackage{cite}
\usepackage{nccmath}
\usepackage{diagbox}
\usepackage{multirow}
\usepackage{cases} 
\usepackage{algorithm}
\usepackage{algpseudocode}
\input arraymac.tex
\makeatletter
\newif\if@restonecol
\makeatother

\DeclareMathOperator*{\argmax}{arg\,max}

\newtheorem{theorem}{Theorem}

\newtheorem{lemma}[theorem]{Lemma}
\begin{document}
\title{Symbol Level Precoding for Systems \\ with Improper Gaussian Interference} \author{Lu~Liu,
Rang~Liu, Ly~V.~ Nguyen,~\IEEEmembership{Member,~IEEE,}\\and~A.~Lee~Swindlehurst,~\IEEEmembership{Fellow,~IEEE}
\thanks{L. Liu was with Center for Pervasive Communications and Computing, University of California, Irvine, USA. She is now with Qorvo Inc., Chandler, USA (e-mail: liul22@uci.edu).}
\thanks{Rang Liu, Ly V. Nguyen and A. Swindlehurst are with Center for Pervasive Communications and Computing, University of California, Irvine, USA (e-mail: \{rangl2, vanln1, swindle\}@uci.edu).}
\thanks{This work was supported by the National Science Foundation under grant CCF-2008714.}}
\maketitle
\begin{abstract}
This paper focuses on precoding design in multi-antenna systems with improper Gaussian interference (IGI), characterized by correlated real and imaginary parts. We first study block level precoding (BLP) and symbol level precoding (SLP) assuming the receivers apply a pre-whitening filter to decorrelate and normalize the IGI. We then shift to the scenario where the base station (BS) incorporates the IGI statistics in the SLP design, which allows the receivers to employ a standard detection algorithm without pre-whitenting. Finally we address the case where the channel and statistics of the IGI are unknown, and we formulate robust BLP and SLP designs that minimize the worst case performance in such settings. Interestingly, we show that for BLP, the worst-case IGI is in fact proper, while for SLP the worst case occurs when the interference signal is maximally improper, with fully correlated real and imaginary parts. Numerical results reveal the superior performance of SLP in terms of symbol error rate (SER) and energy efficiency (EE), especially for the case where there is uncertainty in the non-circularity of the jammer.
\end{abstract}
\section{Introduction}
The evolution from 5G to 6G represents a transformative leap in wireless communication technology, promising unprecedented advancements in speed, capacity, and connectivity. Building on the foundation of 5G's innovations, the development of 6G aims to provide services including ubiquitous mobile ultra broadband (uMUB), ultra-high-speed-with-low-latency communications (uHSLLC), and ultra-high data density (uHDD) \cite{zong20196g,saad2019vision}. Multiple-user (MU) multiple-input multiple-output (MIMO) technology, which involves using multiple antennas to simultaneously serve multiple users or devices, is a cornerstone of both 5G and 6G.

The use of MIMO in ultra-dense networks with smaller cell sizes and more antennas will result in a proportional increase in both inter- and intra-cell interference. To manage the interference, precoding or beamforming is needed to steer the transmit signals towards intended users and mitigate multiuser interference (MUI) \cite{lu2014overview,li2020tutorial}. Existing precoding schemes can be classified as either block level precoding (BLP) or symbol level precoding (SLP). Traditional BLP approaches such as maximum ratio transmission (MRT), zero-forcing (ZF), and optimum interference-constrained or power-constrained precoding \cite{lo1999maximum,spencer2004zero,bjornson2014optimal,tervo2015optimal,schubert2005iterative} are typically designed to suppress the MUI. These approaches assume linear precoding designs, and only depend on the current channel state information (CSI). On the other hand, SLP techniques exploit information about both the CSI and transmitted symbols \cite{alodeh2018symbol, ali2018DPCIR, li2020tutorial, salem2021error}, and generally result in nonlinear designs. Although SLP requires increased complexity, it enables the precoder to take advantage of the MUI and convert it into constructive interference (CI) \cite{masouros2009dynamic,masouros2010correlation,masouros13known}, thus increasing the degrees-of-freedom (DoF) available for precoder design and significantly improving performance.

Numerous SLP approaches have been proposed in recent years. The authors of \cite{masouros2009dynamic} applied SLP using channel inversion precoding, and showed that it enhances the effective signal-to-interference-plus-noise ratio (SINR) at the receivers without investing additional power at the base station (BS). To further improve this design, the technique in \cite{masouros2010correlation} strictly aligns the MUI with the users' desired symbols so that it is converted to CI through a symbol-based correlation rotation matrix that depends on the combined data and CSI. Later, in \cite{alodeh2016energy},  phase alignment constraints were relaxed in order to extend the feasible region of the optimal precoding and achieve additional power savings. All of the above work was developed for either phase-shift keying (PSK) or quadrature-amplitude modulation (QAM) constellations. To tackle the complexity of SLP, an optimal ``closed-form'' SLP solution has been introduced in \cite{li2018interference} to make the design more practical for PSK modulation. Recently, \cite{liu2022end} proposed an end-to-end learning-based approach to optimize the modulation orders for SLP communication systems.

A key strength of SLP is its ability to use CI to move the received signals further away from the symbol decision boundaries, which provides robustness against noise and interference. In the literature, CI regions (CIRs) have been defined to describe the degree to which the received symbols will be robust against noise and unmodeled perturbations. Recent SLP approaches have been designed to maximize the distance between the received signals and the decision boundaries, also referred to as the {\em safety margin} \cite{jedda2017massive,ali2018DPCIR,li2020interference}, for a given maximum transmit power. The methods described in \cite{jedda2018quantized, liu2024robust} introduce Maximum Safety Margin (MSM) precoder designs that are able to minimize an upper bound on the symbol error rate (SER). The MSM approach can be contrasted with algorithms based on minimizing the mean squared-error (MMSE) between the desired and received symbols \cite{masouros13known,masouros2015exploiting}. MSM precoders generally guarantee a better quality-of-service (QoS) for the same level of transmit power, or equivalently the same QoS with less power consumption.

In communication systems, random noise and interference are often modeled as zero-mean Gaussian processes whose real and imaginary parts are uncorrelated and of equal variance \cite{kay1993fundamentals}. Such signals are referred to as ``circular'' or ``proper,'' and are generally more convenient to deal with in transceiver designs \cite{adali2011complex}. However, proper Gaussian signals are not always justified in practice, necessitating adjustments to the transceiver design accordingly \cite{schreier2010statistical}. For example, improper signals can arise from hardware impairments (HWI) caused by phase noise, imperfections in power amplifier manufacturing, non-linearities and I/Q imbalances in RF front ends, etc. \cite{canbilen2018spatial,dao2020iq}. Such factors are more severe in 5G-and-beyond communication systems due to higher carrier frequencies, and thus deserve more attention. Other source of non-circular signals arises from improper complex interference, self-interference, and asymmetric noise \cite{kim2016asym, lameiro2016maximally,alsmadi2018ssk}. 

There are relatively few studies about improper Gaussian interference (IGI) in MIMO downlink communication systems. In \cite{sallem2012optimal}, a maximum likelihood sequence estimator is proposed when the data is corrupted by non-circular zero-mean Gaussian noise in single-input multiple-output (SIMO) systems. In \cite{alsmadi2021effect}, the authors investigated the effects of IGI on quadrature spatial modulation where symbols are expanded into in-phase and quadrature dimensions separately transmitting the real and imaginary parts of an amplitude/phase modulated signal. Instead of viewing IGI as a purely destructive effect, some studies have leveraged the non-circularity of improper Gaussian signaling (IGS) to enhance achievable rates and energy efficiency in different interference channel scenarios. Due to additional DoF offered by IGS, such approaches have been shown to achieve better performance than traditional circular Gaussian signaling \cite{zeng2013optimized, de2018reconsidering, soley2020improper}. In \cite{nasir2019improper}, IGS was designed to manage interference in order to maximize the users'
minimum rate under transmit power constraints. The so-called widely linear precoding has been commonly used to tackle signal non-circularity, and is based on the fact that IGS is a linear transformation of transmit signals and noise, and requires a joint transmitter and receiver design \cite{sterle2007widely,darsena2013widely,zhang2017widely,javed2020journey}.  

In this paper, we focus on precoding design at the transmitter for scenarios involving a jammer transmitting IGI. To begin, assuming the non-circular covariance of the jammer is available at both the BS and the users, we first study the application of an individual pre-whitening filter at each user in order to account for the IGI, and develop modified BLP and SLP approaches. We then study a modification to the SLP design that enables the precoding to be implemented solely by the BS without requiring receiver preprocessing. The non-circular IGI requires the definition of both ``upper'' and ``lower'' safety margins (SMs), which are represented by a bounding box containing  a confidence ellipse for the non-circular noisy observations \cite{strang2022introduction,meyer2023matrix}. The confidence ellipse is centered at the noise-free received signal and defines the region within which noisy received signals will lie with a certain probability. Finally, we consider scenarios where the jammer channel and statistics are unknown, and the BS must design a precoder that is robust to this uncertainty. In particular, we take the conservative approach of designing the BLP and SLP approaches to maximize the performance for the worst-case IGI, and demonstrate that BLP and SLP lead to fundamentally different robust solutions. The main contributions of our work are summarized as follows:
\begin{itemize}
    \item We show how to modify the MMSE-based BLP approach for IGI using pre-whitening at each user, assuming that the CSI and the non-circular covariance of the jamming interference is known. We further demonstrate that when 
    this information is unavailable, the MMSE-based BLP approach should assume a circular interference covariance in order to minimize the worst-case mean-squared error (MSE).
    \item We show how to modify SLP-based designs for IGI using both pre-whitening at each user, and also using transmit-only processing where the receivers do not need to perform any pre-processing prior to detection. 
    \item We further investigate the case when the CSI and non-circularity of the jamming is unknown, and show that, unlike MMSE-based BLP, the worst case for SLP always occurs with maximally improper interference, where the real and imaginary parts of the jammer signal are fully correlated. We then show how to modify the SM-based SLP algorithm to satisfy the worst-case design. 
    \item Finally, we provide comprehensive numerical results to compare MMSE-based BLP with SLP in various settings, particularly in the case where the channel and non-circularity of the IGI is unknown. We also illustrate our theoretical conclusions regarding the worst-case jammer statistics via some graphical examples.
\end{itemize}

\textsl{Notation:} Bold lower case and upper case letters indicate vectors and matrices, and non-bold letters express scalars. The $N\times N$ identity matrix is denoted by $\mathbf{I}_{N}$. $\mathbf{A}_{mn}$ denotes the $(m,n)$-th element in the matrix $\mathbf{A}$, and $a_{m}$ denotes the $m$-th element in the vector $\mathbf{a}$. $(\cdot)^*$, $(\cdot)^{-1}$, $(\cdot)^T$ and $(\cdot)^H$ denote the conjugate, inverse, transpose and Hermitian transpose operators, respectively. $\mathbb{C}^{m\times n}$ ($\mathbb{R}^{m\times n}$) represents the space of complex (real) matrices of dimension $m\times n$. $\mathbb{E}\{\cdot\}$, $\mathbb{P}\{\cdot\}$, $\lvert(\cdot)\rvert$ and $\|\cdot\|$ respectively represent the expectation operator, the probability, the absolute value and the Euclidean norm. $\mathcal{N}(\mu, \sigma^2)$ denotes the normal distribution with mean $\mu$ and variance $\sigma^2$, while $\mathcal{CN}$ denotes the complex normal distribution. The functions $\Tr\{\cdot\}$ and $\text{diag}\{\cdot\}$ respectively indicate the trace of a matrix and a vector composed of the diagonal elements of a square matrix, while $\text{diag}\{\mathbf{a}\}$ denotes a square diagonal matrix with the elements of $\mathbf{a}$ on the main diagonal. A block-diagonal matrix with diagonal blocks $\mathbf{A}_1,\cdots,\mathbf{A}_K$ is denoted by $\text{blockdiag}\{ \, \mathbf{A}_1, \cdots, \mathbf{A}_K\}$. The real and imaginary parts of a complex number are represented by $\mathcal{R}\{\cdot\}$ and $\mathcal{I}\{\cdot\}$, respectively.  For matrices and vectors, $\geq$ and $\leq$ denote element-wise inequalities while $\mathbf{A}\succeq 0$ indicates that the matrix $\mathbf{A}$ is semi-definite positive. Finally, $\triangleq$ means `definded as.'

\section{System Model}\label{sec:systemmodel}
\begin{figure}[!t]
\centering
\includegraphics[width=2.8in]{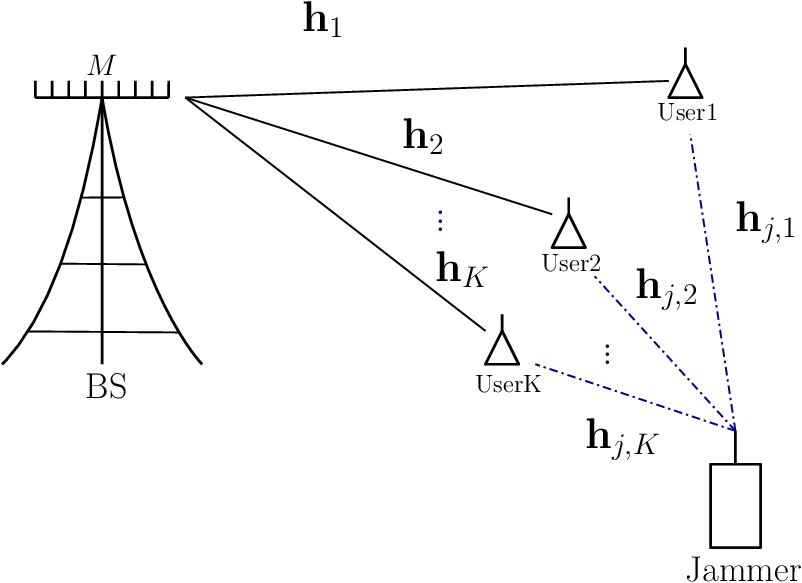}
\caption{System Model}
\label{SystemModel}
\end{figure}
We consider the MIMO downlink system as depicted in Fig.~\ref{SystemModel}, with an $M$-antenna BS, $K$ single-antenna users, and a single-antenna jammer. The signal received at user $k$ can be expressed as
\begin{equation}
y_k=\mathbf{h}_k\mathbf{x}+{h}_{j,k}{z}+n_k\triangleq\mathbf{h}_k\mathbf{x}+c_k\;,
\label{complexYk}
\end{equation}
where $\mathbf{h}_k\in \mathbb{C}^{1\times M}$ and ${h}_{j,k}$ respectively denote the channel from the BS and the jammer to the $k$-th user,  $\mathbf{x}\in \mathbb{C}^{M\times 1}$ represents the precoded signal transmitted by the BS, ${z}$ is a Gaussian interference signal transmitted by the jammer, and $n_k\sim {\cal{CN}}(0, \bar{\sigma}_k^2)$ is circular zero-mean additive Gaussian white noise (AWGN). We combine together the noise and jamming interference as $c_k={h}_{j,k}{z}+n_k$, and refer to this as the {\em effective} noise in the sequel. For a given column vector $\mathbf{a}\in\mathbb{C}^{q\times 1}$ and a given row vector $\mathbf{b}\in\mathbb{C}^{1\times t}$ , we introduce the operators 
\begin{equation*}
\bar{\mathbf{a}}=\begin{bmatrix}\mathcal{R}\{\mathbf{a}\}\\\mathcal{I}\{\mathbf{a}\}\end{bmatrix}\in\mathbb{R}^{2q\times 1}\; , \bar{\mathbf{B}}=\begin{bmatrix}\mathcal{R}\{\mathbf{b}\}&-\mathcal{I}\{\mathbf{b}\}\\\mathcal{I}\{\mathbf{b}\}&\mathcal{R}\{\mathbf{b}\}\end{bmatrix}\in\mathbb{R}^{2\times 2t}\; ,
\end{equation*}
and rewrite Eq.~(\ref{complexYk}) with real-valued quantities
\begin{align*}
\begin{bmatrix}\mathcal{R}\{y_k\}\\\mathcal{I}\{y_k\}\end{bmatrix} &=\begin{bmatrix}\mathcal{R}\{\mathbf{h}_k\}&-\mathcal{I}\{\mathbf{h}_k\}\\\mathcal{I}\{\mathbf{h}_k\}&\mathcal{R}\{\mathbf{h}_k\}\end{bmatrix}\begin{bmatrix}\mathcal{R}\{\mathbf{x}\}\\\mathcal{I}\{\mathbf{x}\}\end{bmatrix}\\
&+\begin{bmatrix}\mathcal{R}\{{h}_{j,k}\}&-{I}\{{h}_{j,k}\}\\\mathcal{I}\{{h}_{j,k}\}&\mathcal{R}\{{h}_{j,k}\}\end{bmatrix}\begin{bmatrix}\mathcal{R}\{{z}\}\\\mathcal{I}\{{z}\}\end{bmatrix}+\begin{bmatrix}\mathcal{R}\{n_k\}\\\mathcal{I}\{n_k\}\end{bmatrix}\;,
\end{align*}
which is equivalently denoted by
\begin{equation}
\bar{\mathbf{y}}_k=\bar{\mathbf{H}}_k\bar{\mathbf{x}}+\bar{\mathbf{H}}_{j,k}\bar{\mathbf{z}}+\bar{\mathbf{n}}_k\label{rxreal}\;.
\end{equation}

We assume that in general the real and imaginary parts of the jammer signal $z$ are correlated; i.e., the jammer signal is noncircular or improper \cite{schreier2010statistical}, and we describe it as follows: 
\begin{equation}
\bar{\mathbf{z}}=\rho\mathbf{T}\begin{bmatrix}v_R\\v_I\end{bmatrix}=\rho\mathbf{T}\bar{\mathbf{v}} \; ,
\label{jammerZ}
\end{equation}
where $\rho>0$ and $\rho^2$ is the jammer transmit power, $\bar{\mathbf{v}}\sim {\cal{N}}(\mathbf{0}, \mathbf{I}_{2})$, and the matrix $\mathbf{T}\in \mathbb{R}^{2\times 2}$ is normalized such that $\mathbf{Q}=\mathbf{T}\mathbf{T}^T$ satisfies $\Tr\{\mathbf{Q}\}=1$. The real-value additive noise is distributed as $\bar{\mathbf{n}}_k\sim{\cal{N}}(\mathbf{0},\frac{1}{2}\bar{\sigma}_k^2\mathbf{I}_2)$. The covariance matrix $\mathbf{Q} =\mathbf{Q}^T \succeq 0$ defines the degree to which the jammer signal $\bar{\mathbf{z}}$ is non-circular. If $\mathbf{Q}=\frac{1}{2} \mathbf{I}_2$, the jammer signal is circular, while unequal diagonal terms or non-zero off-diagonal entries indicate an improper signal. 
With $\bar{\mathbf{c}}_k=\bar{\mathbf{H}}_{j,k}\bar{\mathbf{z}}+\bar{\mathbf{n}}_k$, we can write the covariance of $\bar{\mathbf{c}}_k$ as 
\begin{align}
\mathbb{E}\{\bar{\mathbf{c}}_k\bar{\mathbf{c}}_k^T\}&=\bar{\mathbf{H}}_{j,k}\mathbb{E}\{\bar{\mathbf{z}}\bar{\mathbf{z}}^T\}\bar{\mathbf{H}}_{j,k}^T+\mathbb{E}\{\bar{\mathbf{n}}_k\bar{\mathbf{n}}_k^T\}\\
&=\rho^2\bar{\mathbf{H}}_{j,k}\mathbf{Q}\bar{\mathbf{H}}_{j,k}^T+\frac{1}{2}\bar{\sigma}_k^2\mathbf{I}_2\label{eq:jnGk}\\
&\triangleq \mathbf{G}_k\label{Gk}\; ,
\end{align}
where $\bar{\sigma}_k^2$ is defined as the variance of $n_k$.
Due to the non-circularity of $\bar{z}$, $\mathbf{G}_k$ will in general not be diagonal, and hence conventional precoding techniques that assume circular noise should be reconsidered.

In this downlink scenario, the BS desires to send a symbol $s_k$ to user $k$ for $k=1,\cdots,K$. In BLP with circular noise, the transmitted signal would be expressed as a linear function of the complex symbols $\mathbf{x}=\mathbf{P}_c\mathbf{s}$, where $\mathbf{s}=\left[ s_1 \; \cdots \; s_K \right]^T$ is the symbol vector and $\mathbf{P}_c$ is the $M\times K$ precoder. In the non-circular case, we separate the real and imaginary parts and write the linear BLP transmit signal as $\bar{\mathbf{x}}=\mathbf{P}\bar{\mathbf{s}}$. In general,
\begin{equation*}
\mathbf{P}\neq\begin{bmatrix}\mathcal{R}\{\mathbf{P}_c\}&-\mathcal{I}\{\mathbf{P}_c\}\\\mathcal{I}\{\mathbf{P}_c\}&\mathcal{R}\{\mathbf{P}_c\}\end{bmatrix}\; ,
\end{equation*}
hence we need to modify the precoder for non-circular noise. For the case of SLP, the transmit signal $\mathbf{x}$, or more generally $\bar{\mathbf{x}}$, is a non-linear function of $\mathbf{s}$ or $\bar{\mathbf{s}}$, respectively. For simplicity, we will assume that the elements of $\mathbf{s}$ are PSK symbols, although the methods can be generalized to other (e.g., QAM) constellations, using approaches similar to those in \cite{jedda2018quantized}.

\section{Pre-Whitening Methods}
If each user $k$ has knowledge of its own interference-plus-noise covariance $\mathbf{G}_k$, then a straightforward way of dealing with improper interference is through a pre-whitening step that decorrelates and normalizes the real and imaginary parts of the interference. This prewhitening also impacts the BS-user channels, so the BS will also need to know $\mathbf{G}_k$ for each user, presumably through a feedback channel from the users. Note that each $\mathbf{G}_k$ is only a $2\times 2$ matrix, so the amount of feedback required is minimal. In this section we describe how to implement BLP and SLP with non-circular interference pre-whitening at each user. While straightforward, the results derived here will be useful in later sections of the paper.

\subsection{Block Level Precoding}\label{sec:smartBLP}
There is limited research on precoding methods that address improper noise. In this section, we use the MMSE criterion to develop a benchmark BLP approach. According to Eqs.~(\ref{rxreal})-(\ref{Gk}), if $\mathbf{G}_k$ is known at the receivers, the non-circular disturbance $\bar{\mathbf{c}}_k$ can be pre-whitened as follows:
\begin{equation}
\mathbf{G}_k^{-\frac{1}{2}}\bar{\mathbf{y}}_k=\mathbf{G}_k^{-\frac{1}{2}}\bar{\mathbf{H}}_k\bar{\mathbf{x}}+\mathbf{G}_k^{-\frac{1}{2}}\bar{\mathbf{c}}_k\label{BLPyk}\; .
\end{equation}
Using the following definitions for all users:
\begin{eqnarray}
    \begin{bmatrix}
      y_{E,k}^1\\ y_{E,k}^2 
    \end{bmatrix}=\mathbf{G}_k^{-\frac{1}{2}}\bar{\mathbf{y}}_k\; ,
    \begin{bmatrix}
      \mathbf{h}_{E,k}^1\\ \mathbf{h}_{E,k}^2 
    \end{bmatrix}=\mathbf{G}_k^{-\frac{1}{2}}\bar{\mathbf{H}}_k\; ,
    \begin{bmatrix}
      c_{E,k}^1\\ c_{E,k}^2 
    \end{bmatrix}=\mathbf{G}_k^{-\frac{1}{2}}\bar{\mathbf{c}}_k\; ,\nonumber
\end{eqnarray}
the received symbols can be expressed as
\begin{equation}
\mathbf{y}_E=\mathbf{H}_E\mathbf{P}\bar{\mathbf{s}}+\mathbf{c}_E\; ,\label{yEbar}
\end{equation}
where 
\begin{eqnarray*}   \mathbf{y}_E=\begin{bmatrix}y_{E,1}^1&\cdots&y_{E,K}^1&y_{E,1}^2&\cdots&y_{E,K}^2\end{bmatrix}^T\; ,\\
\mathbf{H}_E=\begin{bmatrix}{\mathbf{h}_{E,1}^1}^T&\cdots&{\mathbf{h}_{E,K}^1}^T&{\mathbf{h}_{E,1}^2}^T&\cdots&{\mathbf{h}_{E,K}^2}^T\end{bmatrix}^T\; ,\\
\mathbf{c}_E=\begin{bmatrix}c_{E,1}^1&\cdots&c_{E,K}^1&c_{E,1}^2&\cdots&c_{E,K}^2\end{bmatrix}^T\; .
\end{eqnarray*}

For a given $\mathbf{H}_E$, we consider the problem of minimizing the mean 
squared error (MMSE) as the criterion to design the BLP \cite{sterle2007widely, darsena2013widely, zhang2017widely}, which is formulated as
\begin{align}
     \min_{\mathbf{P},\beta}\quad& \mathbb{E}\{\|\beta^{-1}\mathbf{y}_E-\bar{\mathbf{s}}\|^2\}\label{opt:smartBLP}\\
     \text{subject to}\quad&\mathbb{E}\{\|\mathbf{P}\bar{\mathbf{s}}\|^2\}\leq P_t\; ,
\end{align}
where $\beta$ is a scaling factor, and $P_t$ is the transmit power budget. The corresponding Lagrangian function is 
\begin{align*}
    \mathcal{L}(\mathbf{P}, \beta,\lambda)&=\mathbb{E}\{\|\beta^{-1}\mathbf{y}_E-\bar{\mathbf{s}}\|^2\}+\lambda(\mathbb{E}\{\|\mathbf{P}\bar{\mathbf{s}}\|^2\}-P_t)\\
    &=K-\frac{1}{2}\beta^{-1}\Tr\{\mathbf{H}_E\mathbf{P}+\mathbf{P}^T\mathbf{H}_E^T\}\\
    &\quad+\frac{1}{2}\beta^{-2}\Tr\{\mathbf{H}_E\mathbf{P}\mathbf{P}^T\mathbf{H}_E^T\}+2K\beta^{-2}\\
    &\quad+\frac{1}{2}\lambda\Tr\{\mathbf{P}\mathbf{P}^T\}-\lambda P_t\; ,
\end{align*}
where $\lambda\geq 0$ denotes the Lagrange multiplier, and $\mathbb{E}\{\bar{\mathbf{s}}\bar{\mathbf{s}}^T\}=\frac{1}{2}\mathbf{I}_{2K}$, $\mathbb{E}\{\mathbf{c}_E\mathbf{c}_E^T\}=\mathbf{I}_{2K}$. 

The solution to (\ref{opt:smartBLP}) should satisfy the matrix equations
\begin{align}
    \frac{d \mathcal{L}(\mathbf{P}, \beta,\lambda)}{d \mathbf{P}}&=-\beta^{-1}\mathbf{H}_E^T+\beta^{-2}\mathbf{H}_E^T\mathbf{H}_E\mathbf{P}+\lambda\mathbf{P}=\mathbf{0}_{2M\times 2K }\; ,\label{LagrangianP}\\
    \frac{d \mathcal{L}(\mathbf{P}, \beta,\lambda)}{d \beta}&=\beta^{-2}\Tr\{\mathbf{H}_E\mathbf{P}\}-\beta^{-3}\Tr\{\mathbf{H}_E\mathbf{P}\mathbf{P}^T\mathbf{H}_E^T\}\nonumber\\
    &\quad\quad-4K\beta^{-3}=0\label{LagrangianBeta}\; ,
\end{align}
which ultimately yields
\begin{eqnarray}
    \mathbf{P}&=&\beta\boldsymbol{\Delta}\mathbf{H}_E^T\label{BLPprecoding}\; ,\\
    \beta&=&\sqrt{\frac{2P_t}{\Tr\{\boldsymbol{\Delta}\mathbf{H}_E^T\mathbf{H}_E\boldsymbol{\Delta}\}}} \; , \label{BLPbeta}
\end{eqnarray}
where $\boldsymbol{\Delta}=(\mathbf{H}_E^T\mathbf{H}_E+a\mathbf{I}_{2M})^{-1}$ and $a=\frac{2K}{P_t}$. Since this BLP method takes into account the non-circularity of the interference, we refer to it as ``pre-whitened BLP'' or PW-BLP.

\subsection{Symbol Level Precoding}\label{sec:MSM}
As discussed in \cite{jedda2018quantized}, the MSM SLP approach is designed such that the noise-free received signal for user $k$ will be located within the CIR of the transmitted symbol with a safety margin of $\delta_k$. For the case of circular $c_k$, the choice of $\delta_k$ is made based on a certain desired level of reliability given the variance of the total noise plus interference seen at the receiver, which we denote as 
\begin{equation}
    \sigma_k^2 = \Tr\{\Gbf_k\} = \rho^2 |h_{j,k}|^2 +\bar{\sigma}_k^2 \; .
\end{equation}
The value of $\delta_k$ can, for example, be chosen to minimize an upper bound on the symbol error probability (SEP) \cite{jedda2018quantized}. To maintain consistency with the definition of $\delta_k$ in the case of non-circular interference, we pre-whiten user $k$'s signal as follows:
\begin{equation}
\gamma_k\mathbf{G}_k^{-\frac{1}{2}}\bar{\mathbf{y}}_k=\gamma_k\mathbf{G}_k^{-\frac{1}{2}}\bar{\mathbf{H}}_k\bar{\mathbf{x}}+\gamma_k\mathbf{G}_k^{-\frac{1}{2}}(\bar{\mathbf{H}}_{j,k}\bar{\mathbf{z}}+\bar{\mathbf{n}}_k) \; , \label{eq:rxrealwhite}
\end{equation}
where $\gamma_k$ is a scaling factor chosen to ensure that the total power of the jammer plus noise remains the same as in the circular noise case, i.e., 
\begin{equation}
\mathbb{E}\left\{\Tr\left\{\gamma_k^2\mathbf{G}_k^{-\frac{1}{2}}\bar{\mathbf{c}}_k\bar{\mathbf{c}}_k^T(\mathbf{G}_k^{-\frac{1}{2}})^T\right\}\right\}=\sigma_k^2\; .
\end{equation}
Thus, the scaling should be chosen as $\gamma_k=\frac{\sigma_k}{\sqrt{2}}$.

If we define 
\begin{equation}
\gamma_k\mathbf{G}_k^{-\frac{1}{2}}\bar{\mathbf{H}}_k\triangleq\begin{bmatrix}\mathbf{h}_{e1,k}\\\mathbf{h}_{e2,k}\end{bmatrix}\; ,
\end{equation}
and the effective {\em complex} channel after whitening is denoted as $\mathbf{h}_{e,k}$,
then the received noiseless symbol in Eq.~(\ref{eq:rxrealwhite}) can be rewritten in real-valued form as
\begin{equation}
\gamma_k\mathbf{G}_k^{-\frac{1}{2}}\bar{\mathbf{H}}_k\bar{\mathbf{x}}=\begin{bmatrix}\mathbf{h}_{e1,k}\bar{\mathbf{x}}\\\mathbf{h}_{e2,k}\bar{\mathbf{x}}\end{bmatrix}=\begin{bmatrix}\mathcal{R}\{\mathbf{h}_{e,k}\mathbf{x}\}\\\mathcal{I}\{\mathbf{h}_{e,k}\mathbf{x}\}\end{bmatrix}\; .
\end{equation}

As an example, for unit-magnitude $D$-PSK signals $s_{k}\in\{s|s=\exp(j\pi(2d-1)/D),\ d\in\{1,\cdots, D\}\}$, the safety margin can be calculated as \cite{jedda2018quantized} 
\begin{equation}
    \delta_k=\mathcal{R}\{s_k^*\mathbf{h}_{e,k}\mathbf{x}\}\sin\theta-\lvert\mathcal{I}\{s_k^*\mathbf{h}_{e,k}\mathbf{x}\}\rvert\cos\theta\label{safetymargin}\; ,
\end{equation}
where $\theta = \pi/D$.
Furthermore, to express the SM more clearly, we derive
\begin{align*}
\mathcal{R}\{s_k^*\mathbf{h}_{e,k}\mathbf{x}\}&=\mathcal{R}\{s_k^*\}\mathcal{R}\{\mathbf{h}_{e,k}\mathbf{x}\}-\mathcal{I}\{s_k^*\}\mathcal{I}\{\mathbf{h}_{e,k}\mathbf{x}\}\\
&=\mathcal{R}\{s_k^*\}\mathbf{h}_{e1,k}\bar{\mathbf{x}}-\mathcal{I}\{s_k^*\}\mathbf{h}_{e2,k}\bar{\mathbf{x}}\\
&=\tilde{\mathbf{h}}_{e,k}^-\bar{\mathbf{x}}\; ,
\end{align*}
\begin{align*}
\mathcal{I}\{s_k^*\mathbf{h}_{e,k}\mathbf{x}\}&=\mathcal{I}\{s_k^*\}\mathcal{R}\{\mathbf{h}_{e,k}\mathbf{x}\}+\mathcal{R}\{s_k^*\}\mathcal{I}\{\mathbf{h}_{e,k}\mathbf{x}\}\\
&=\mathcal{I}\{s_k^*\}\mathbf{h}_{e1,k}\bar{\mathbf{x}}+\mathcal{R}\{s_k^*\}\mathbf{h}_{e2,k}\bar{\mathbf{x}}\\
&=\tilde{\mathbf{h}}_{e,k}^+\bar{\mathbf{x}}\; ,
\end{align*}
where  $\tilde{\mathbf{h}}_{e,k}^-\bar{\mathbf{x}}=\mathcal{R}\{s_k^*\}\mathbf{h}_{e1,k}-\mathcal{I}\{s_k^*\}\mathbf{h}_{e2,k}$, and $\tilde{\mathbf{h}}_{e,k}^+\bar{\mathbf{x}}=\mathcal{I}\{s_k^*\}\mathbf{h}_{e1,k}+\mathcal{R}\{s_k^*\}\mathbf{h}_{e2,k}$.
Then, one possible SLP optimization problem is to maximize the smallest safety margin over all $K$ users for a given transmit power budget, as follows:
\begin{align*}
\max_{\bar{\mathbf{x}}}\quad& \delta\\
\text{subject to}\quad&
(\tilde{\mathbf{h}}_{e,k}^-\sin \theta-\tilde{\mathbf{h}}_{e,k}^+\cos \theta)\bar{\mathbf{x}}\geq\delta,\ \forall k\; ,\\
&(\tilde{\mathbf{h}}_{e,k}^-\sin \theta+\tilde{\mathbf{h}}_{e,k}^+\cos \theta)\bar{\mathbf{x}}\geq \delta,\ \forall k\; ,\\
&\|\bar{\mathbf{x}}\|^2\leq P_t\; .
\end{align*}
This approach is essentially equivalent to the conventional MSM method in \cite{jedda2018quantized}, with the only difference being the pre-whitening transformation of the channel to account for the non-circularity of the effective noise. An alternative formulation is to minimize the transmit power subject to given safety margin constraints $\delta_1^0,\cdots , \delta_K^0$ for each user:
\begin{align*}
\min_{\bar{\mathbf{x}}}\quad& \|\bar{\mathbf{x}}\|^2\\
\text{subject to}\quad&
(\tilde{\mathbf{h}}_{e,k}^-\sin \theta-\tilde{\mathbf{h}}_{e,k}^+\cos \theta)\bar{\mathbf{x}}\geq\delta^0_k,\ \forall k\; ,\\
&(\tilde{\mathbf{h}}_{e,k}^-\sin \theta+\tilde{\mathbf{h}}_{e,k}^+\cos \theta)\bar{\mathbf{x}}\geq \delta^0_k,\ \forall k
\; .
\end{align*}
We will refer to the above algorithms as pre-whitened SLP, or simply PW-SLP, and our numerical results will illustrate the performance gain achieved by adjusting the conventional SLP approach in scenarios with IGI.

\section{SLP with Transmit-Only Processing}\label{sec:smartSLP}
The SLP design of the previous section requires not only the precoding but also receiver processing for pre-whitening the received signal, which motivates the question of whether we can implement an IGI-aware SLP design solely at the transmitter by directly optimizing the SM and CIR. In this section we show that if the covariance matrix $\mathbf{G}_k$ of the jammer signal at each user is known, the definition of SM in \cite{jedda2018quantized,liu2024robust} can be adjusted to appropriately account for the IGI. Intuitively, IGI will produce an elliptical rather than a circular cloud around the noiseless received signal in the complex plane, and for modulations such as PSK and QAM, this will require the definition of two safety margins in order to guarantee the user QoS, defined as the likelihood that the received signal lies in the CIR. We will employ confidence ellipses \cite{strang2022introduction, meyer2023matrix} to formulate the SM and CIR. 

\subsection{Confidence Ellipse}
\begin{figure}
\centering
\begin{subfigure}[b]{0.27\textwidth}
    \centering
\includegraphics[width=\textwidth]{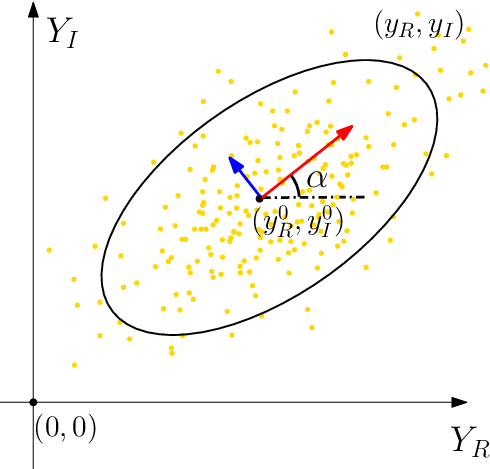}
\caption{}
\label{fig:ellipseOrigin}
\end{subfigure}
\hfill
\begin{subfigure}[b]{0.3\textwidth}
   \centering  
\includegraphics[width=\textwidth]{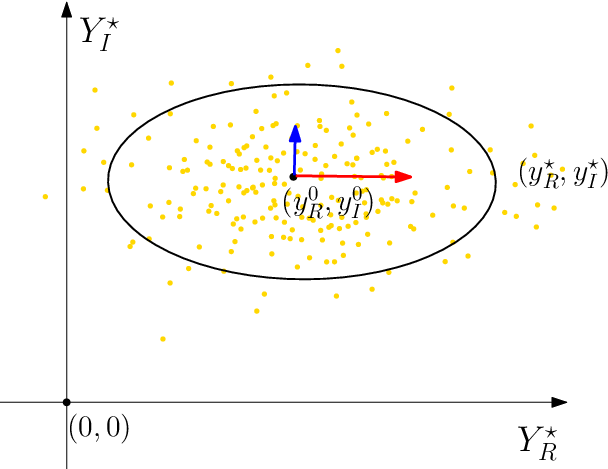}
\caption{}
\label{fig:ellipseRotate}
\end{subfigure}
\caption{(a) Confidence ellipse with correlated 2D data samples;\ (b) Confidence ellipse after rotation to decorrelate the 2D data samples.}
\end{figure}

Given normally distributed two-dimensional (2D) data $\{Y_R,Y_I\}$, where the two variables are correlated with covariance matrix $\mathbf{G}$, we can draw a confidence ellipse to define a region that contains data samples $(y_R,y_I)$ with a preset confidence value. For example, in Fig.~\ref{fig:ellipseOrigin}, the center of the ellipse is $(y^0_R,y^0_I)$, where $y^0_R=\mathbb{E}\{Y_R\}$ and $\ y^0_I=\mathbb{E}\{Y_I\}$ are the mean values of the two variables. The orientation of the ellipse is denoted by the angle $\alpha\ (0\leq\alpha<180^{\circ})$ between the major axis of the ellipse and the $Y_R$-axis. When $\mathbf{G}$ is diagonal, we have $\alpha=0$. The eigenvectors of $\mathbf{G}$ correspond to the directions of the major and minor axes of the ellipse as depicted by the red and blue arrows in Fig.~\ref{fig:ellipseOrigin}, while the square root of the eigenvalues $\lambda_{1},\ \lambda_{2}$ of $\mathbf{G}$ correspond to the spread of data in the two directions. It is easy to show that
\begin{equation}
\alpha=\arctan\frac{v_2}{v_1}
\label{AlphaEllipse}
\end{equation}
where $\mathbf{v}=[ v_1 \; v_2]^T$ is the principle eigenvector of  $\mathbf{G}$ corresponding to the largest eigenvalue \cite{strang2022introduction,meyer2023matrix}.

\begin{figure}[!t]
\centering
\includegraphics[width=3in]{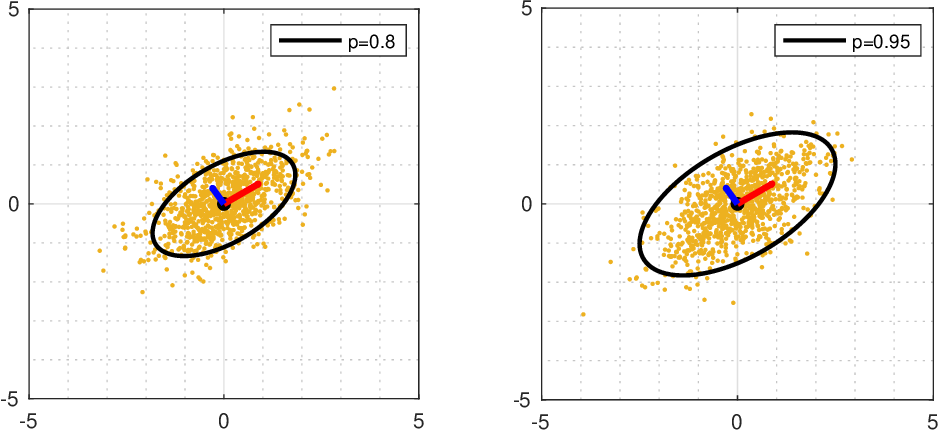}
\caption{The preset $p$ decides the size of confidence ellipses.}
\label{fig:pElipse}
\end{figure}

The size of the ellipse depends on the chosen confidence level $0 \le p \le 1$, which represents the asymptotic fraction of the data samples located inside the ellipse, as illustrated in Fig.~(\ref{fig:pElipse}). To mathematically express the relationship governing $p$, we first decorrelate $Y_R$ and $Y_I$ by rotating the data samples with respect to the ellipse center by the angle $-\alpha$:
\begin{equation}
\mathbf{R}=\begin{bmatrix}\cos\alpha&\sin\alpha\\-\sin\alpha&\cos\alpha\end{bmatrix}\label{rotationMatrix}\; .
\end{equation}
The rotated data samples in Fig.~\ref{fig:ellipseRotate} are defined by
\begin{equation}
\begin{bmatrix}y_R^{\star}\\y_I^{\star}\end{bmatrix}=\mathbf{R}\begin{bmatrix}y_{R}\\y_{I}\end{bmatrix}\; ,
\end{equation}
and the covariance matrix for $\{Y_R^{\star},Y_I^{\star}\}$,  
\begin{equation}
\mathbf{G}^{\star}=\mathbf{R}\mathbf{G}\mathbf{R}^H\label{newCovariance}\; ,
\end{equation}
is diagonal, with eigenvalues equal to $\lambda_1,\lambda_2$ due to the orthogonality of $\mathbf{R}$. Since $Y_R^{\star}$ and $Y_I^{\star}$ are independent normally distributed random variables, according to the $\chi^2$ distribution we can express the confidence ellipse as 
\begin{equation}
\mathbb{P}\left\{\frac{(y_R^{\star}-y^0_R)^2}{\lambda_1}+\frac{(y_I^{\star}-y^0_I)^2}{\lambda_2}\leq\omega\right\}=p\; ,
\end{equation}
where $\omega=-2\ln(1-p)$. 

\subsection{Constructive Interference Region for IGI}
\begin{figure}[!t]
\centering
  \includegraphics[width=3in]{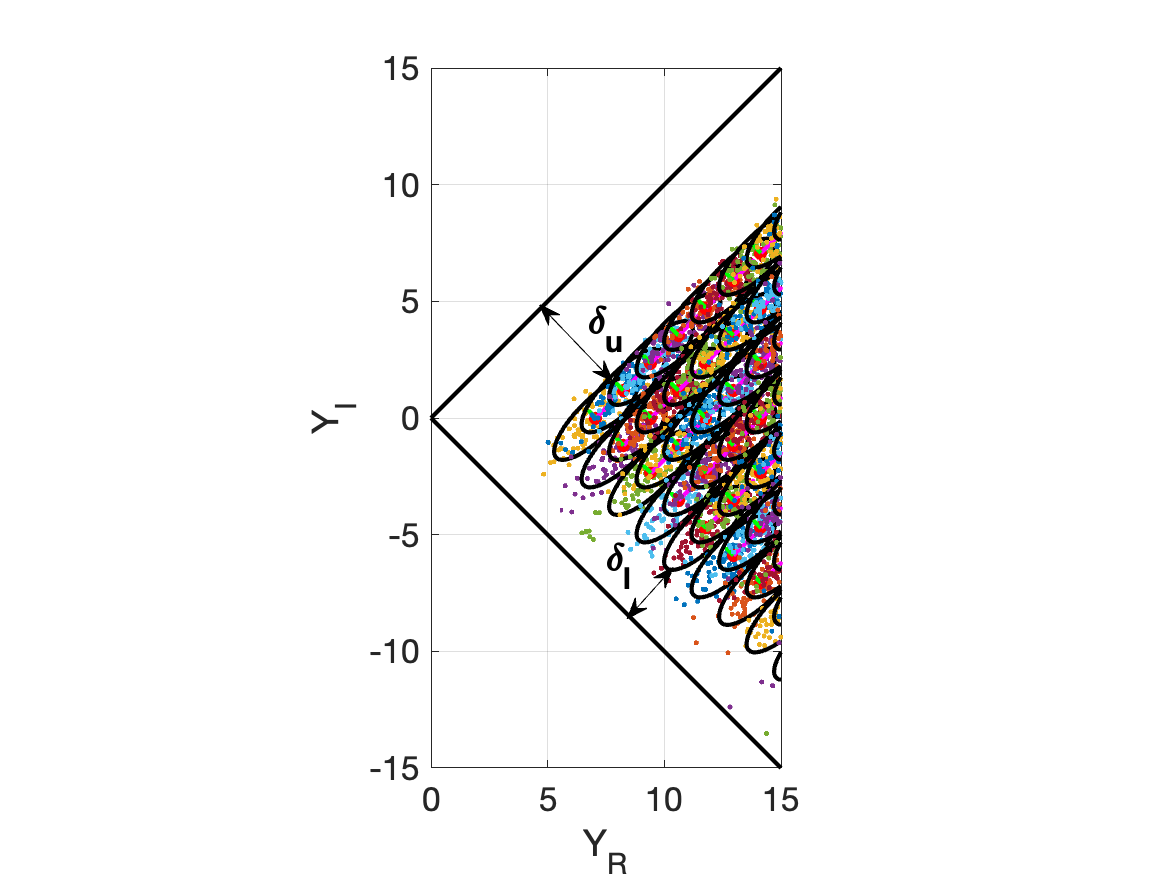}
  \caption{CIR with IGI added to noise-free received signals.}
  \label{DecisionBoundary}
\end{figure}
To describe the CIR in a unified way, we take the case of PSK as an example, and we rotate the original coordinate system for the constellation by the negative phase of the desired symbol, i.e., $\measuredangle s_k^*$, to obtain the modified coordinate system in Fig.~\ref{DecisionBoundary}. We can see that with IGI, the safety margin must be determined separately for the two symbol decision boundaries; in particular, we define the minimum distance from the confidence ellipses to the two decision boundaries as the upper safety margin ($\delta_u^0$) and lower safety margin ($\delta_l^0$), where $\delta_u^0\neq \delta_l^0$. In Fig.~\ref{EllipseLine}, we plot one of the ellipses, with the center at the $k$th user's noiseless received signal $(\bar{\mathbf{H}}_{k}^1\bar{\mathbf{x}},\bar{\mathbf{H}}_{k}^2\bar{\mathbf{x}})$ in the modified coordinate system, where
\begin{align}
\bar{\mathbf{H}}_{k}^1&=\begin{bmatrix}\mathcal{R}\{{s_k^*\mathbf{h}}_k\}&-\mathcal{I}\{{s_k^*\mathbf{h}}_k\}\end{bmatrix}\; ,\\
\bar{\mathbf{H}}_{k}^2&=\begin{bmatrix}\mathcal{I}\{{s_k^*\mathbf{h}}_k\}&\mathcal{R}\{{s_k^*\mathbf{h}}_k\}
\end{bmatrix}\; .
\end{align}
The closest point on the ellipse to each decision boundary, $y_{\mathcal{I}}=\pm (\tan\theta)y_{\mathcal{R}}$, is where the tangent line of the ellipse is parallel to the corresponding boundary, as illustrated in Fig.~\ref{EllipseLine}, where $\theta=\frac{\pi}{D}$ for $D$-PSK. In the discussion below, we show how to obtain expressions for the SMs $\delta_{u,k}^0$ and $\delta_{l,k}^0$ in the presence of IGI, so that the SLP approach can be implemented.

\begin{figure}[!t]
\centering
  \includegraphics[width=2.5in]{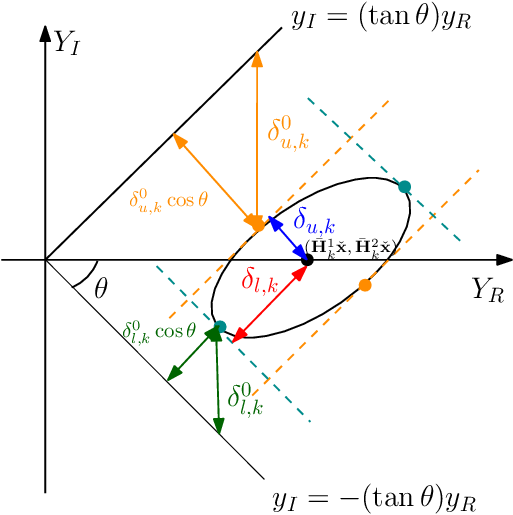}
  \caption{Distance from confidence ellipse to decision boundaries.}
  \label{EllipseLine}
\end{figure}

\subsection{SLP for Non-Circular Interference}
We will develop the SLP precoding problem to minimize the BS transmit power while achieving the users' various SM constraints. In the modified coordinate system, we obtain the real and imaginary parts of the interference signal as 
\begin{equation}
    \begin{bmatrix}
        \mathcal{R}\{s_k^*c_k\}\\
        \mathcal{I}\{s_k^*c_k\}
        \end{bmatrix}=\bar{\mathbf{S}}_k^T\bar{\mathbf{c}}_k\; ,
\end{equation}
where $\bar{\mathbf{S}}_k=\begin{bmatrix}\mathcal{R}\{s_k\}&-\mathcal{I}\{s_k\}\\\mathcal{I}\{s_k\}&\mathcal{R}\{s_k\}\end{bmatrix}$, and its covariance matrix is given by
\begin{equation}
\mathbb{E}\left\{\bar{\mathbf{S}}_k^T\bar{\mathbf{c}}_k\left(\bar{\mathbf{S}}_k^T\bar{\mathbf{c}}_k\right)^T\right\}=\bar{\mathbf{S}}_k^T\mathbf{G}_k\bar{\mathbf{S}}_k\triangleq \check{\mathbf{G}}_k\; .\label{checkGk}
\end{equation}
We obtain the orientation angle of the confidence ellipse using Eq.~(\ref{AlphaEllipse}):
\begin{equation}
\check{\alpha}_k=\arctan\frac{\check{v}_{k,2}}{\check{v}_{k,1}}\; , \label{eq:alphakcheck}
\end{equation}
where $\check{\mathbf{v}}_k = [ \check{v}_{k,1} \; \check{v}_{k,1} ]^T$ is the eigenvector of $\check{\mathbf{G}}_k$ corresponding to the largest eigenvalue. 

To ensure that the confidence ellipse is inside the decision region, we use two lines to constrain its location to satisfy
\begin{align}
y_k^R\sin\theta-y_k^I\cos\theta&\geq 0 \; ,\\
y_k^R\sin\theta+y_k^I\cos\theta&\geq 0\; , 
\end{align}
and we set constraints for the four points on the ellipse where the slope of the tangent line is either $\tan\theta$ or $-\tan\theta$. These points will include the two closest points to the decision boundaries, as illustrated in Fig.~\ref{EllipseLine}. The four lines defined by these points form a bounding box that contains the ellipse. As a result, two SM constraints will be necessary to describe the location of the confidence ellipse with respect to the decision boundaries, which we refer to as the {\em upper} and {\em lower} SM constraints. In Appendix~\ref{appendix:upperSM} we show that the upper SM constraint is defined as follows:
\begin{align}
(\bar{\mathbf{H}}_{k}^1\sin\theta-\bar{\mathbf{H}}_{k}^2\cos\theta)\bar{\mathbf{x}}&\geq\delta_{u,k}+\delta_{u,k}^0\cos\theta\; ,\label{upperSM1}\\
(\bar{\mathbf{H}}_{k}^1\sin\theta-\bar{\mathbf{H}}_{k}^2\cos\theta)\bar{\mathbf{x}}&\geq -\delta_{u,k}+\delta_{u,k}^0\cos\theta \; ,\label{uppserSM2}
\end{align}
where 
\begin{align}
    |\delta_{u,k}|&=|u_k\sin\theta+v_k\cos\theta| \label{eq:delta_u}\\
    &=\sqrt{\omega_k}\sqrt{\lambda_{1,k}\sin^2(\check{\alpha}_k-\theta)
+\lambda_{2,k}\cos^2(\check{\alpha}_k-\theta)}\;\label{eq:deltaUpperRobust} 
\end{align}
represents the distance from the noise-free received signal to the tangent line of the bounding box that intersects the orange points in Fig.~\ref{EllipseLine}. This is a critical dimension, since it represents how much the interference can perturb the received signal and still remain in the confidence ellipse defined by $p$. In general, the larger this distance, the more power that will be required to satisfy the given safety margin.

Using a similar derivation, for the green points on the ellipse in Fig.~\ref{EllipseLine} where the slope of the tangent line is $-\tan\theta$, we obtain the {\em lower} SM constraints
\begin{align}
(\bar{\mathbf{H}}_{k}^1\sin\theta+\bar{\mathbf{H}}_{k}^2\cos\theta)\bar{\mathbf{x}}&\geq\delta_{l,k}+\delta_{l,k}^0\cos\theta\; ,\label{lowerSM1}\\
(\bar{\mathbf{H}}_{k}^1\sin\theta+\bar{\mathbf{H}}_{k}^2\cos\theta)\bar{\mathbf{x}}&\geq -\delta_{l,k}+\delta_{l,k}^0\cos\theta\; ,\label{lowerSM2}
\end{align}
where $\delta_{l,k}^0$ denotes the preset lower SM, 
\begin{align}  
|\delta_{l,k}|&=|u_k^{\prime}\sin\theta-v_k^{\prime}\cos\theta|\label{eq:delta_l}\\
&=\sqrt{\omega_k}\sqrt{\lambda_{1,k}\sin^2(\check{\alpha}_k+\theta)
+\lambda_{2,k}\cos^2(\check{\alpha}_k+\theta)}\; ,\label{eq:deltaLowerRobust}
\end{align}
and $u_k^{\prime}$, $v_k^{\prime}$ can be found using $\kappa_k^{\prime}=\tan(\check{\alpha}_k+\theta)$ and steps similar to those in Appendix \ref{appendix:upperSM}. The value of $\delta_{l,k}$ represents the distance from the noise-free received signal to the tangent lines intersecting the green points, and like $\delta_{u,k}$, a larger $\delta_{l,k}$ generally means a higher required transmit power.

If we denote
\begin{equation}
\mathbf{A}_k^-=\bar{\mathbf{H}}_{k}^1\sin\theta-\bar{\mathbf{H}}_{k}^2\cos\theta\; ,\label{eq:Aneg}
\end{equation}
Eq.~(\ref{upperSM1}) and (\ref{uppserSM2}) can be combined as
\begin{equation}
\mathbf{A}_k^-\bar{\mathbf{x}}-\delta_{u,k}^0\cos\theta\geq|\delta_{u,k}|\; .
\end{equation}
Similarly, Eq.~(\ref{lowerSM1}) and (\ref{lowerSM2}) can be combined to yield 
\begin{equation}
\mathbf{A}_k^+\bar{\mathbf{x}}-\delta_{l,k}^0\cos\theta\geq|\delta_{l,k}|\; ,
\end{equation}
where 
\begin{equation} \mathbf{A}_k^+=\bar{\mathbf{H}}_{k}^1\sin\theta+\bar{\mathbf{H}}_{k}^2\cos\theta\; .\label{eq:Apos}
\end{equation}

\begin{algorithm}[!t]
\begin{small}
    \caption{Non-Circular Symbol Level Precoding}
\label{alg1}
\begin{algorithmic}[1]
    \State {\textbf{Input:}
    $\mathbf{h}_k$, $h_{j,k}$, $s_k$, $\delta_{u,k}^0$, $\delta_{l,k}^0$, $\sigma_k$, $\omega_k$, $\forall k$, $\rho$, $\mathbf{Q}$, and $\theta$.}
    \State {\textbf{Output:} $\bar{\mathbf{x}}^{\text{nc}}$.}  
         \For {$k \in \{1,\cdots,K\}$}
            \State Calculate $\mathbf{A}_k^-$ in (\ref{eq:Aneg}) , and $\mathbf{A}_k^+$ in (\ref{eq:Apos}).
            \State Calculate $\mathbf{G}_k$ using (\ref{Gk}), and its eigenvalues $\lambda_{1,k},\ \lambda_{2,k}$.
             \State Calculate $\check{\mathbf{G}}_k$ using (\ref{checkGk}, and $\check{\alpha}_k$ using (\ref{eq:alphakcheck}).
            \State Calculate $\delta_{u,k}$ using (\ref{eq:deltaUpperRobust}) and $\delta_{l,k}$ using (\ref{eq:deltaLowerRobust}).           
         \EndFor
    \State Solve the optimization problem (\ref{opt:smartSLP}) to obtain the optimal precoded vector $\bar{\mathbf{x}}^{\text{nc}}$.
\end{algorithmic}
\end{small}
\end{algorithm}

One possible SLP optimization problem that results from the derivations above is to design the signal $\bar{\mathbf{x}}$ that minimizes the transmit power at the BS and satisfies preset upper/lower safety margins at each user:
\begin{align}
\min_{\bar{\mathbf{x}}}\quad&\bar{\mathbf{x}}^H\bar{\mathbf{x}}\label{opt:smartSLP}\\
\text{subject to}\quad&
\mathbf{A}_k^-\bar{\mathbf{x}}-\delta_{u,k}^0\cos\theta\geq|\delta_{u,k}|\; ,\ \forall k\; ,\\
&\mathbf{A}_k^+\bar{\mathbf{x}}-\delta_{l,k}^0\cos\theta\geq|\delta_{l,k}|\; ,\ \forall k\; .
\end{align}
This is a quadratic programming problem with linear inequality constraints that can be efficiently solved using standard numerical methods. An alternative formulation of the problem is possible in which the smaller of the upper and lower SMs is maximized for a given total transmit power. In either case, we refer to this approach as ``non-circular SLP,'' or simply NC-SLP. The steps required to implement the algorithm are listed in Algorithm.~ \ref{alg1}.

In contrast, when the BS ignores the non-circularity of the interference and implements the conventional SLP approach assuming circular interference, we will refer to this approach as ``naive SLP.'' For a fair comparison, in implementing the naive SLP approach we will assume that the safety margin is set assuming the same total power as in the case of IGI, i.e., $\sigma_k^2 = \Tr\{\mathbf{G}_k\}$. 

\section{Robust Non-Circular Precoding}\label{sec:robust}
The methods discussed above assume that the covariance matrix $\mathbf{G}_k$ of the effective noise at user $k$ is known at both the users and the BS. In this section, we drop the assumption of known $\mathbf{G}_k$, although we assume that the power level of the AWGN $\bar{\sigma}_k^2$ and the jammer $\rho^2 |h_{j,k}|^2$ are still known. Removing the assumption of known $\mathbf{G}_k$ eliminates the need for instantaneous jammer CSI, and enables modeling of an intelligent jammer that changes $\mathbf{Q}$ from time to time to confuse the legitimate receivers. Thus, it is important that robust algorithms be designed to handle cases with unknown $\mathbf{G}_k$. Our goal in this section is to design precoders that are robust to knowledge of the jammer channel and statistics. In Section~\ref{sec:robustBLP} we address robustness for the case of BLP, and in Section~\ref{sec:robustSLP} we do the same for SLP.

\subsection{Robust BLP}\label{sec:robustBLP}
To make the MMSE BLP approach described in Section~\ref{sec:smartBLP} robust to knowledge of the interference covariances $\mathbf{G}_k$ at the users, we minimize the worst-case MSE over all possible choices of $\mathbf{G}_k$. The optimization problem can be formulated as
\begin{align}
    \min_{\mathbf{P},\beta} \;  \Big\{\max_{\mathbf{G}_k\in \mathcal{S}_k} \; & \mathbb{E}\{\|\beta^{-1}\mathbf{y}_E-\bar{\mathbf{s}}\|^2\}\Big\}\label{opt:robustBLP}
    \\
     \text{subject to}\quad&\mathbb{E}\{\|\mathbf{P}\bar{\mathbf{s}}\|^2\}\leq P_t\; ,\nonumber
\end{align} 
where $\mathcal{S}_k = \left\{ \mathbf{G} \, : \, \mathbf{G}\succeq 0 \, , \, \Tr\{\mathbf{G}\}=\sigma_k^2\right\}$.
As a first step, in the lemma below we prove that for MMSE BLP, the worst-case $\mathbf{G}_k$ for $k=1,\cdots,K$ actually corresponds to circular jamming.
\begin{lemma}\label{lemma:robustBLP}
The maximum MSE 
\begin{equation}\label{eq:lemma1MMSE}
    \max_{\mathbf{G}_k\in \mathcal{S}_k}\; \mathbb{E}\{\|\beta^{-1}\mathbf{y}_E-\bar{\mathbf{s}}\|^2\}
\end{equation}
is achieved with $\mathbf{G}_k=\frac{\sigma^2_k}{2} \mathbf{I}_2$.
\end{lemma}
\begin{proof}
See Appendix \ref{appendix: robustBLP}.
\end{proof}
Interestingly, the robust BLP design is implemented under the assumption that the effective noise is uncorrelated in the real and imaginary part, which is equivalent to assuming $\mathbf{G}_k=\frac{1}{2}(\rho^2|h_{j,k}|^2+\bar{\sigma}_k^2)\mathbf{I}_2$. Substituting this value for $\mathbf{G}_k$ into Eq.~(\ref{BLPprecoding}) and (\ref{BLPbeta}) and implementing the MMSE precoder will be referred to hereafter as ``robust BLP". 

\subsection{Robust SLP}\label{sec:robustSLP}
Here we study the robustness of SLP to knowledge of the IGI covariance $\mathbf{G}_k$. Due to the higher complexity of finding the SLP solution, this is a challenging problem. Fortunately, we will show below that only the extreme cases where the degree of non-circularity is maximum need be considered to find the worst-case scenario. Thus, only rank-deficient $\mathbf{Q}$ need be considered, and we show how the worst-case $\mathbf{Q}$ can be found via a one-dimensional search.

As with the robust MMSE BLP approach, we will design the SLP algorithm to minimize the worst-case transmit power with respect to all possible 
$\mathbf{G}_k\in\mathcal{S}_k$. This can be achieved by reformulating the non-circular SLP problem in~(\ref{opt:smartSLP}) as follows:
\begin{align}
\min_{\bar{\mathbf{x}}}& \;\bar{\mathbf{x}}^T\bar{\mathbf{x}}\label{opt:robustSLP2}\\
\text{subject to}\quad& 
\mathbf{A}_k^-\bar{\mathbf{x}}-\delta_{u,k}^0\cos\theta\geq\max_{\mathbf{G}_k\in \mathcal{S}_k} |\delta_{u,k}|\; ,\ \forall k\; ,\nonumber\\
&\mathbf{A}_k^+\bar{\mathbf{x}}-\delta_{l,k}^0\cos\theta\geq\max_{\mathbf{G}_k\in \mathcal{S}_k}|\delta_{l,k}|\;  ,\ \forall k\; .\nonumber  
\end{align}
In other words, for each $k$, we find the worst-case $\mathbf{G}_k$ that maximizes the size of the confidence ellipse and pushes it closer to the decision boundaries. The closer any of these points is to the decision boundary, the larger the transmit power required to satisfy the desired safety margin. Unlike~(\ref{opt:smartSLP}), this is a challenging problem; as it stands, (\ref{opt:robustSLP2}) requires solving a convex SLP problem for every possible $\mathbf{G}_k$ to find the one that requires the most transmit power. 

Fortunately, the complexity of the problem can be significantly reduced based on the following lemma.
\begin{lemma}\label{lemma:robustSLP}
   The matrix $\mathbf{G}_k = \rho^2\bar{\mathbf{H}}_{j,k}\mathbf{Q}\bar{\mathbf{H}}_{j,k}^T+\frac{1}{2}\bar{\sigma}_k^2\mathbf{I}_2$ corresponding to the worst-case (maximum) transmit power in~(\ref{opt:robustSLP2}) will have a rank-deficient $\mathbf{Q}$. 
\end{lemma}
\begin{proof}
As mentioned above, the transmit power $\bar{\mathbf{x}}^H\bar{\mathbf{x}}$ required to satisfy the constraints in~(\ref{opt:robustSLP2}) will increase when the worst case SM thresholds $|\delta_{u,k}|$ and $|\delta_{l,k}|$ increase. Hence, to find the maximum transmit power, we first find the $\mathbf{G}_k$, or equivalently the $\mathbf{Q}$, that maximizes the SM thresholds.
Based on $\mathbf{G}_k$ in Eq.~(\ref{eq:jnGk}), if $\lambda_1^Q$ and $\lambda_2^Q$ are eigenvalues of $\mathbf{Q}$,
we will obtain
\begin{align}
&\lambda_{1,k}=\rho^2|h_{j,k}|^2\lambda_1^Q+\frac{1}{2}\sigma_k^2\; ,\label{lambda1k}\\
&\lambda_{2,k}=\rho^2|h_{j,k}|^2\lambda_2^Q+\frac{1}{2}\sigma_k^2\; ,\label{lambda2k}
\end{align}
where $\lambda_1^Q+\lambda_2^Q=1$. Then the term under the square root of Eq.~(\ref{eq:deltaUpperRobust}) can be written as
\begin{align}
p_k\triangleq&\lambda_{1,k}\sin^2(\check{\alpha}_k-\theta)
+\lambda_{2,k}\cos^2(\check{\alpha}_k-\theta)\nonumber\\
=&\rho^2|h_{j,k}|^2\left(\lambda_{1}^Q\sin^2(\check{\alpha}_k-\theta)
+\lambda_{2}^Q\cos^2(\check{\alpha}_k-\theta)\right)+\frac{1}{2}\sigma_k^2\nonumber\\
=&\rho^2|h_{j,k}|^2\lambda_{1}^Q\left(\sin^2(\check{\alpha}_k-\theta)
-\cos^2(\check{\alpha}_k-\theta)\right)\nonumber\\
&\quad\quad+\rho^2|h_{j,k}|^2\cos^2(\check{\alpha}_k-\theta)+\frac{1}{2}\sigma_k^2\label{lambda_neg}\;.
\end{align}
When $\check{\alpha}_k$ is fixed, $p_k$ is a linear function of $\lambda_1^Q$. Since $0\leq\lambda_1^Q\leq 1$, we can find the maximum value of $p_k$ as follows:
\begin{itemize}
    \item  When $\sin^2(\check{\alpha}_k-\theta)
\geq\cos^2(\check{\alpha}_k-\theta)$, we have
\begin{equation}
p_k^{u1}\triangleq\max\ p_k=\rho^2|h_{j,k}|^2\sin^2(\check{\alpha}_k-\theta)
+\frac{1}{2}\sigma_k^2\; ,\label{eq:maxpu1}
\end{equation}
with $\lambda_1^Q=1,\ \lambda_2^Q=0$;\\
\item When $\sin^2(\check{\alpha}_k-\theta)
\leq\cos^2(\check{\alpha}_k-\theta)$, we have
\begin{equation}
p_k^{u2}\triangleq\max\ p_k=\rho^2|h_{j,k}|^2\cos^2(\check{\alpha}_k-\theta)
+\frac{1}{2}\sigma_k^2\; ,\label{eq:maxpu2}
\end{equation}
with $\lambda_1^Q=0,\ \lambda_2^Q=1$.
\end{itemize}

Similarly, for the lower SM constraints in Eq.~(\ref{eq:deltaLowerRobust}), if $q_k\triangleq\lambda_{1,k}\sin^2(\check{\alpha}_k+\theta)
+\lambda_{2,k}\cos^2(\check{\alpha}_k+\theta)\nonumber$, we can obtain the maxima
\begin{equation}
p_k^{l1}\triangleq\max\ q_k=\rho^2|h_{j,k}|^2\sin^2(\check{\alpha}_k+\theta)
+\frac{1}{2}\sigma_k^2\label{eq:maxpl1}
\end{equation}
or
\begin{equation}
p_k^{l2}\triangleq\max\ q_k=\rho^2|h_{j,k}|^2\cos^2(\check{\alpha}_k+\theta)
+\frac{1}{2}\sigma_k^2\; .\label{eq:maxpl2}
\end{equation}

Thus, to achieve the maximum value of $|\delta_{u,k}|$ in (\ref{eq:deltaUpperRobust}) or the maximum value of $|\delta_{l,k}|$ in (\ref{eq:deltaLowerRobust}), either $\lambda_1^Q$ or $\lambda_2^Q$ should be zero, which means the matrix $\mathbf{Q}$ will be rank deficient. \qedsymbol
\end{proof}

\begin{algorithm}[!t]
\begin{small}
    \caption{Robust Symbol Level Precoding}
\label{alg2}
\begin{algorithmic}[1]
    \State {\textbf{Input:}
    $\mathbf{h}_k$, $h_{j,k}$, $s_k$, $\delta_{u,k}^0$, $\delta_{l,k}^0$, $\sigma_k$, $\omega_k$,  $\forall k$, $\rho$, $\theta$, and $N_{div}$.}
    \State {\textbf{Output:} $\bar{\mathbf{x}}^{\text{robust}}$.}  
    \For {$n \in \{1,\cdots,N_{div}\}$}
         \State {Calculate $\phi^n$ using (\ref{eq:betaN}).}
         \State {Calculate $\mathbf{v}^n$, the principle eigenvector of $\mathbf{Q}$ by (\ref{eq:vn}).}
         \For {$k \in \{1,\cdots,K\}$}
            \State {Calculate $\mathbf{A}_k^-$ in (\ref{eq:Aneg}) , and $\mathbf{A}_k^+$ in (\ref{eq:Apos}).}
            \State {Calculate $\check{\mathbf{v}}_{k}^n$ using (\ref{eq:v_k_n}).}
             \State {Calculate $\check{\alpha}_k^n$ using (\ref{eq:alpha_k_n})}. 
            \State Calculate $p_k^{u1}(n)$, $p_k^{u2}(n)$, $p_k^{l1}(n)$, $p_k^{l2}(n)$ using (\ref{eq:maxpu1})-(\ref{eq:maxpl2}).           
         \EndFor
    \State Solve the optimization problem (\ref{opt:robustSLPfast}) to obtain the optimal precoded vector $\bar{\mathbf{x}}^{n\star}$, as well as the transmit power $|\bar{\mathbf{x}}^{n\star}|^2$.
    \EndFor
    \State{Return $\bar{\mathbf{x}}^{\text{robust}}=\bar{\mathbf{x}}^{n}$, where $n=\argmax_{n\star} {|\bar{\mathbf{x}}^{n\star}|^2}$.}
\end{algorithmic}
\end{small}
\end{algorithm}

Lemma~\ref{lemma:robustSLP} stands in contrast to Lemma~\ref{lemma:robustBLP}; while a circular interference signal generates the worst-case MMSE for BLP, for SLP it is a maximally non-circular jammer signal with a rank deficient $\mathbf{Q}$ that leads to the worst performance. To solve~\eqref{opt:robustSLP2}, we must evaluate the worst case $|\delta_{u,k}|$ and $|\delta_{l,k}|$ for the orientations $\check{\alpha}$ that result from all possible rank-deficient~$\mathbf{Q}$. 

Note that if the improper covariance $\mathbf{Q}$ is rank deficient, then $\bar{\mathbf{H}}_{j,k}\Q\bar{\H}^T_{j,k}$ is also improper and rank deficient, so finding the worst-case rank-one $\Q$ is equivalent to finding the worst-case rank-one $\tilde{\Q}=\bar{\H}_{j,k}\Q\bar{\H}_k^T$. Thus,
we express the principle eigenvector of $\mathbf{Q}$ with direction $\phi$ as
 \begin{equation}
     \mathbf{v}=\begin{bmatrix}
    \cos{\phi}\\
    \sin{\phi}
\end{bmatrix}\; ,\label{eq:vn}
\end{equation}
from which we can use Eq.~(\ref{Gk}), (\ref{newCovariance}), and (\ref{checkGk}), to get the principle eigenvector of $\check{\mathbf{G}}_k$ as
\begin{equation}
\check{\mathbf{v}}_{k}=\begin{bmatrix}
    \cos{\check{\alpha}_k}\\
    \sin{\check{\alpha}_k}
\end{bmatrix}\propto\frac{\rho\bar{\mathbf{S}}_k^T\bar{\mathbf{H}}_{j,k}}{\rho |h_{j,k}|}\mathbf{v}\; ,\label{eq:v_k_n}  
\end{equation}
and hence
\begin{equation}
\check{\alpha}_k=\arctan\frac{\check{\mathbf{v}}_k(2)}{\check{\mathbf{v}}_k(1)}=\phi+\measuredangle h_{j,k}-\measuredangle s_k = \tilde{\phi} - \measuredangle s_k \label{eq:alpha_k_n} \; ,
\end{equation}
where $\tilde{\phi}=\phi+\measuredangle h_{j,k}$. Thus, in principle we must search over all possible $\tilde{\phi}$ to find the worst case non-circularity for each user. For the results presented in the next section, we simply discretize $\tilde{\phi}$ over $N_{div}$ points in $(0,\pi]$:
\begin{equation}
   \tilde{\phi}^n=\frac{n\pi}{N_{div}},\ n\in \{1,\cdots,N_{div}\}\; . \label{eq:betaN}
\end{equation}
For each direction $\tilde{\phi}^n$, we find $\check{\alpha}_k^n$ for each user and substitute this value into (\ref{eq:maxpu1})-(\ref{eq:maxpl2}) to find the worst case bounds for the SLP optimization:
\begin{align}
\bar{\mathbf{x}}^{n\star} = \arg\min_{\bar{\mathbf{x}}^n}&\; {(\bar{\mathbf{x}}}^n)^T\bar{\mathbf{x}}^n\label{opt:robustSLPfast}\\
\text{subject to}\quad
&\mathbf{A}_k^-\bar{\mathbf{x}}^n-\delta_{u,k}^0\cos\theta\geq \sqrt{\omega_k}\sqrt{p_k^{u1}(n)}\; ,\ \forall k\; ,\nonumber\\
&\mathbf{A}_k^-\bar{\mathbf{x}}^n-\delta_{u,k}^0\cos\theta\geq \sqrt{\omega_k}\sqrt{p_k^{u2}(n)}\; ,\ \forall k\; ,\nonumber\\
&\mathbf{A}_k^+\bar{\mathbf{x}}^n-\delta_{l,k}^0\cos\theta\geq \sqrt{\omega_k}\sqrt{p_k^{l1}(n)}\; ,\ \forall k\; ,\nonumber\\
&\mathbf{A}_k^+\bar{\mathbf{x}}^n-\delta_{l,k}^0\cos\theta\geq \sqrt{\omega_k}\sqrt{p_k^{l2}(n)}\; ,\ \forall k\; . \nonumber
\end{align}
The solution is then the $\bar{\mathbf{x}}^{n\star}$ for $n=1,\cdots,N_{div}$ with largest norm. Algorithm~\ref{alg2} above outlines the detailed steps to implement this ``robust SLP'' approach. We will show in the next section that only a very small value for $N_{div}$ is necessary to find the desired solution.

\section{Numerical Results}
In this section we use Monte Carlo simulations over 1000 independent channel realizations in each trial to assess the performance of the proposed BLP and SLP approaches. Table~\ref{table:summary} summarizes the BLP and SLP approaches that will be investigated. For each channel realization, a block of 200 symbols is transmitted. The channels $\mathbf{h}_k$ and $h_{j,k}$ are composed of i.i.d. Gaussian random variables with zero mean and unit variance. The power of the AWGN $n_k$ is assumed to be $\bar{\sigma}=1$ (if not otherwise specified) for all users, and the probability $p_k=p$ defining the confidence ellipses is identical for all users. We also employ the same safety margin ($\delta$) for all users with respect to a given SNR threshold $\psi$, where the relationship between $\delta$ and $\psi$ can be expressed as \cite{li2018interference}
\begin{equation*}
    \psi=\frac{\delta^2}{(\sin\theta)^2(\rho^2+\sigma^2)}\; .
\end{equation*}

\begin{table}[!t]
\centering
\caption{Summary of Precoding Approaches}
\begin{tabular}{|c|c|c|} 
\hline
Name&$\mathbf{G}_k$&Criterion\\
\hline
PW-MSM&Available&MSM\\
\hline
PW-SLP&Available&Minimum Power\\
\hline
Naive BLP&Ignore&MMSE\\
\hline
Naive SLP&Ignore&Minimum Power\\
\hline
PW-BLP&Available&MMSE\\
\hline
NC-SLP&Available&Minimum Power\\
\hline
Robust BLP&$\frac{\sigma^2_k}{2} \mathbf{I}_2$&MMSE\\
\hline
Robust SLP&Rank-deficient&Minimum Power\\
\hline
\end{tabular}
\vspace{1ex}
\label{table:summary}
\end{table}

In some cases we will use the energy efficiency (EE) to quantify the power-performance trade-off of different designs, which is defined as the ratio of the throughput $\tau$ to the transmit power per channel:
\begin{equation*}
    \text{EE}=\frac{\tau}{T\times \|\bar{\mathbf{x}}_c\|^2}, 
\end{equation*} 
where
$\tau$ is defined in \cite{salem2021error, liu2024robust} as:
\begin{equation*}\label{throughput}
\tau=(1-P_B)\times c\times T\times K , 
\end{equation*}
where $P_B$ is the block error rate (BLER),  $c=\log_2D$ is the number of bits per modulation symbol, $T$ is the block length and $K$ is the number of receivers. 

 





\begin{figure}[!t]
\centering
\includegraphics[width=3in]{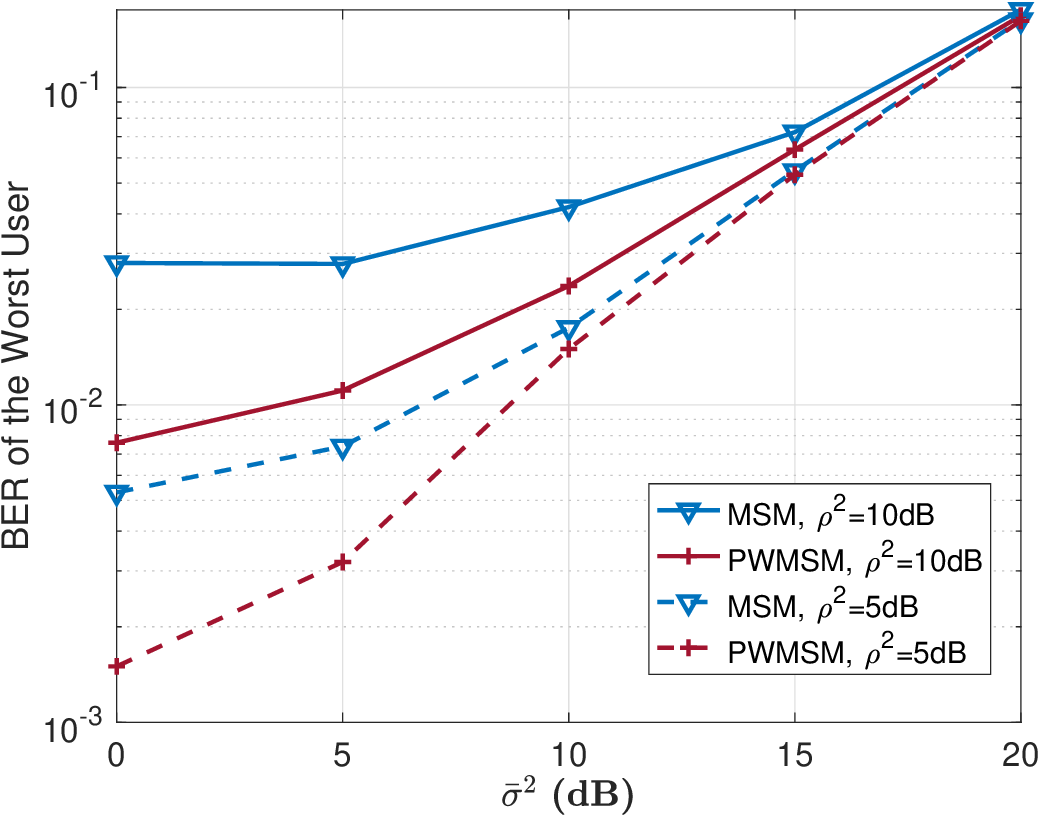}
\caption{BER of the worst user vs. AWGN standard deviation $\bar{\sigma}$, with $M=K=8$, $P_t=30$dB, QPSK.}
\label{fig:MSMwithJammerBER}
\end{figure}
\subsection{Impact of Improper Noise}
First, we explore the impact of improper noise on the MSM algorithm \cite{jedda2017massive,jedda2018quantized}. We compare it with the PW-MSM approach in Section~\ref{sec:MSM}, which considers the IGI and employs a pre-whitening transformation to decorrelate the interference. Fig.~\ref{fig:MSMwithJammerBER} shows that PW-MSM achieves a lower BER, especially at lower SNR where the impact of the interference is greater. Moreover, as power of the jammer increases from 5 to 10 dB, the gap between the two approaches becomes wider. As $\bar{\sigma}$ increases, the impact of the non-circular jammer signal diminishes, and the effective noise becomes more circular. Thus, for high AWGN power, the PW-MSM algorithm converges to the traditional MSM approach, as expected. 

\begin{figure}[!t]
\centering
\includegraphics[width=3in]{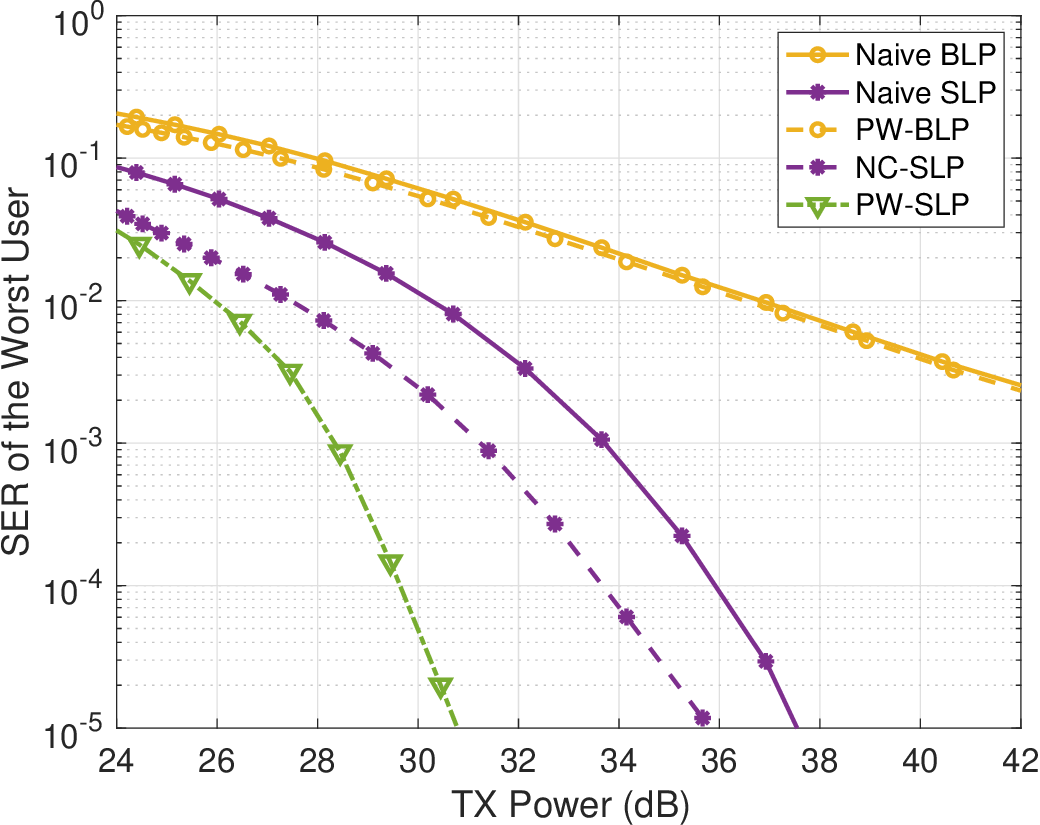}
\caption{SER of the worst user vs. different SNR thresholds $\psi$ with $\rho^2=10$dB, $M=K=8$, $p=80\%$, QPSK.}
\label{fig:PWNCworstSER}
\end{figure}

\begin{figure}[!t]
\centering
\includegraphics[width=3in]{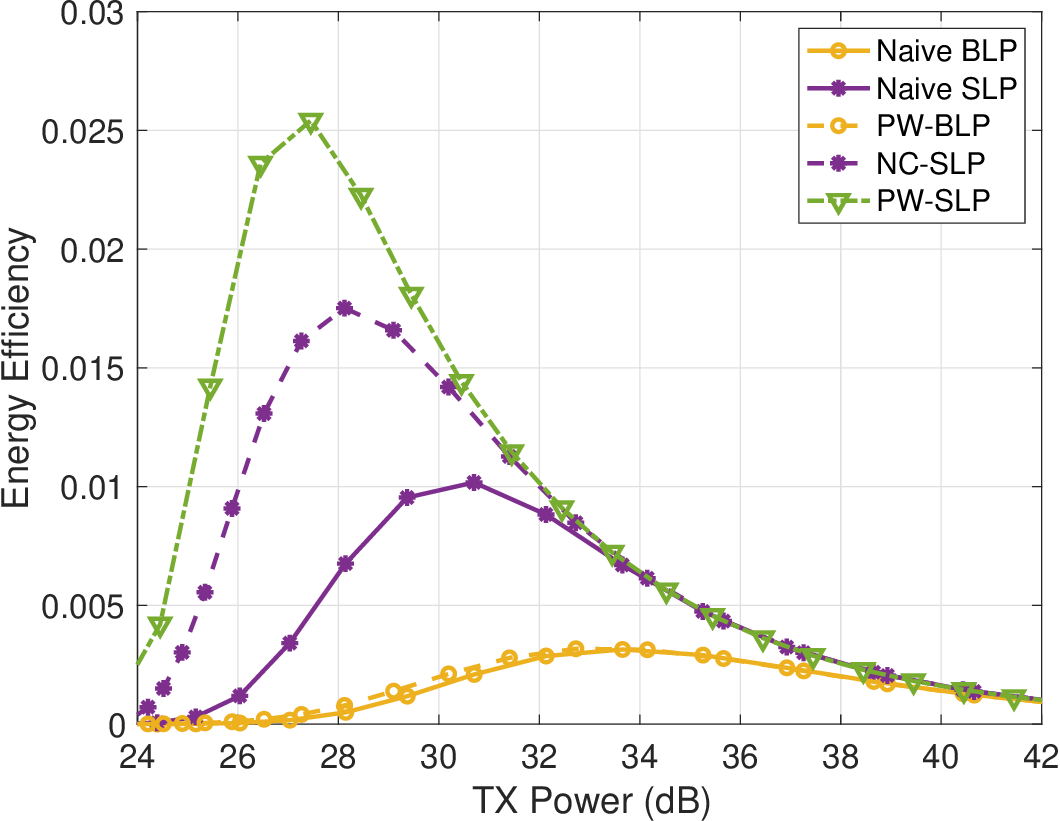}
\caption{Energy efficiency vs. different transmit powers with $\rho^2=10$dB, $M=K=8$, $p=80\%$, QPSK.}
\label{fig:PWNCEE}
\end{figure}

Fig.~\ref{fig:PWNCworstSER} illustrates the SER of the worst user when the BLP/SLP approaches are designed with or without taking the IGI into account. Generally, SLP techniques achieve significantly lower SER for the worst user compared to BLP designs. Even the naive SLP approach that ignores the non-circular noise is superior to PW-BLP that takes the non-circular noise into account. The advantage of SLP is particularly pronounced in the high SNR region, where the proposed SLP designs can provide more than 10dB improvement compared to BLP. The proposed PW-SLP algorithm that requires pre-whitening at the receiver performs the best, and can achieve up to a 6dB improvement compared to naive SLP, demonstrating the significance of redesigning the precoding to take the IGI into account. If pre-whitening at the receiver is not possible, the proposed NC-SLP algorithm can be applied, which also exhibits significant improvement in the presence of IGI. In Fig.~\ref{fig:PWNCEE}, we plot the energy efficiency versus the transmit power for each approach. The SLP designs are clearly more energy efficient compared to the BLP algorithms, and their EE is further improved when the IGI is taken into account.

\subsection{Impact of Probability Constraint $p$}

\begin{figure}[!t]
\centering
\includegraphics[width=3in]{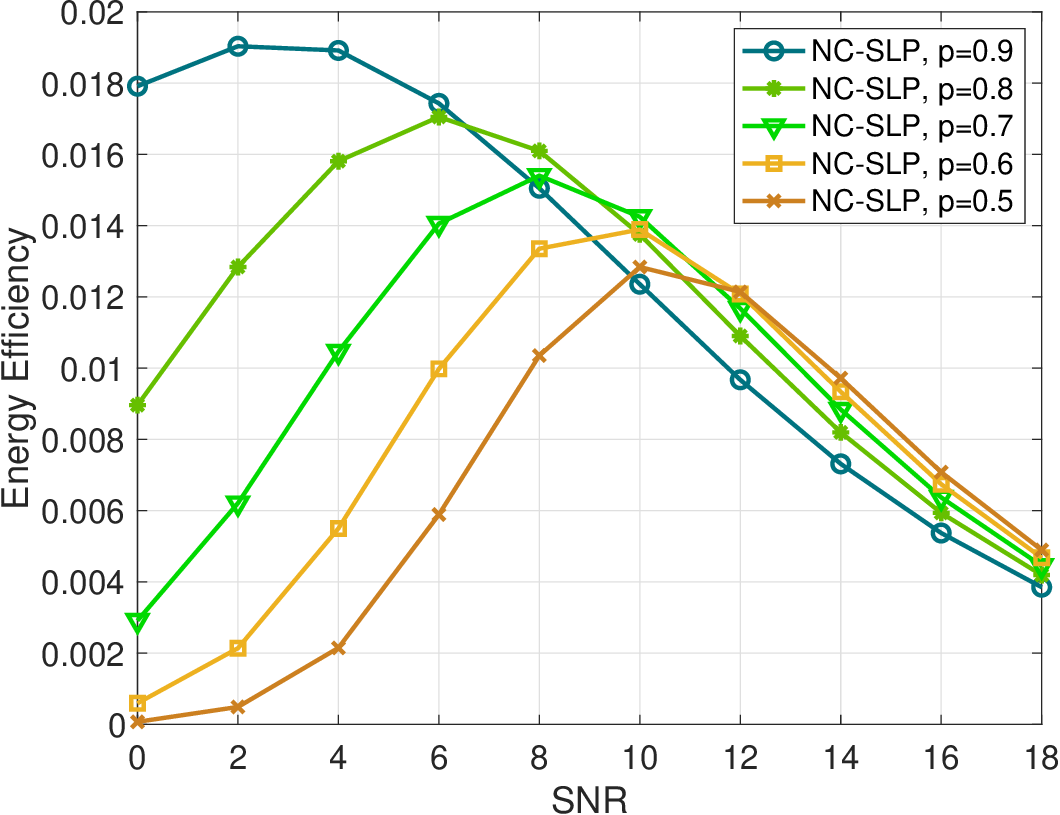}
\caption{The effect of probability constraints on energy efficiency when using NC-SLP, with $\rho^2=10$dB, $M=K=8$, QPSK.}
\label{probSmartEE}
\end{figure}

It is evident from Section~\ref{sec:smartSLP} that the probability constraint plays an important role in SLP design. Fig.~\ref{probSmartEE} shows the EE of the NC-SLP approach under different probability settings. For low SNRs, a larger probability constraint improves the EE of NC-SLP, but at high SNRs, a lower probability is superioer. Interestingly, as $p$ increases from 0.5 to 0.9, the transition point where the EE begins to deteriorate shifts towards lower SNRs. This insight is valuable for power budgeting. For instance, when $p=0.8$, the optimal SNR constraint to achieve the best EE is approximately 6dB, beyond which increasing the transmit power does not provide sufficient growth in throughput.

\subsection{Performance of Robust Designs}
In the first part of this section, we will examine the performance of BLP and SLP as a function of the jammer circularity $\Q$ to validate the results of Section~\ref{sec:robust}. To do so, we will plot the performance (either MSE or transmit power) as a function of $q_{11}$ and $q_{12}$, the defining elements of 
\begin{equation}
    \Q = \left[ \begin{array}{cc} q_{11} & q_{12} \\ q_{12} & 1-q_{11} \end{array} \right] \; ,
\end{equation}
under the positive definiteness constraint $(q_{11}-\frac{1}{2})^2+q_{12}^2 \le \frac{1}{4}$.

\begin{figure}[!t]
\centering
\includegraphics[width=3in]{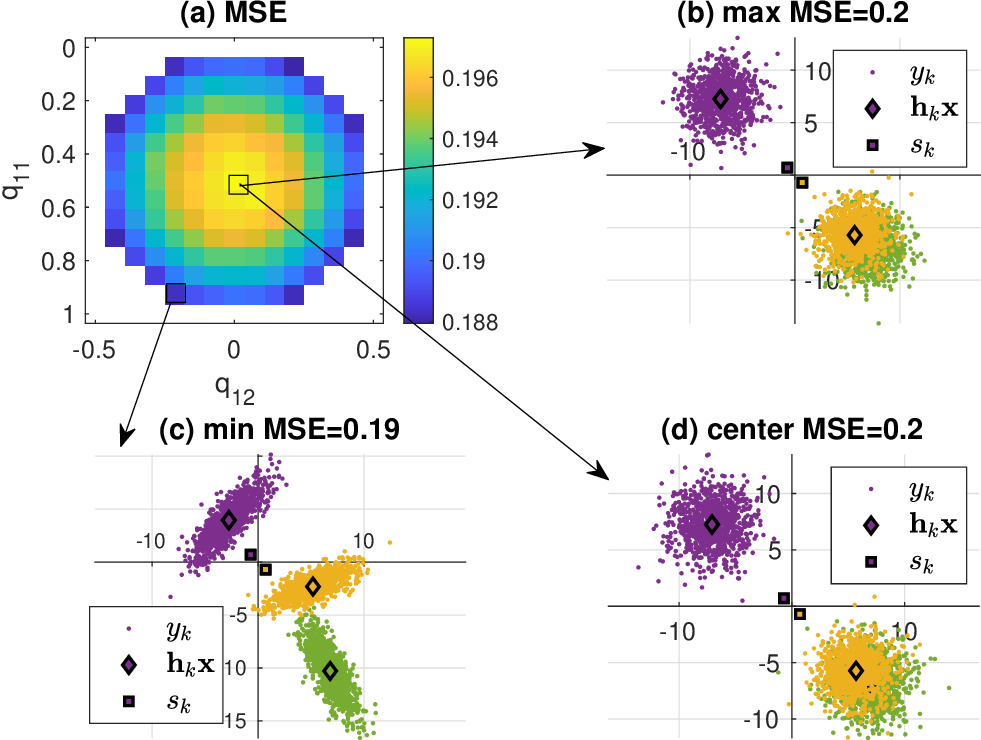}
\caption{PW-BLP based on MMSE, $P_t=20$dB, $\rho^2=10$dB, $M=K=3$, QPSK. (a) MSE for various $q_{11}$ and $q_{12}$; plots of the desired symbol ($s_k$), received noiseless signals ($\mathbf{h}_k\mathbf{x}$), received signals plus improper noise ($y_k$) when (b) MSE is maximum; (c) MSE is  minimum; (d) MSE is calculated with $\mathbf{Q}=\frac{1}{2}\mathbf{I}_2$.}
\label{robustBLPmap}
\end{figure}

\begin{figure}[!t]
\centering
\includegraphics[width=3in]{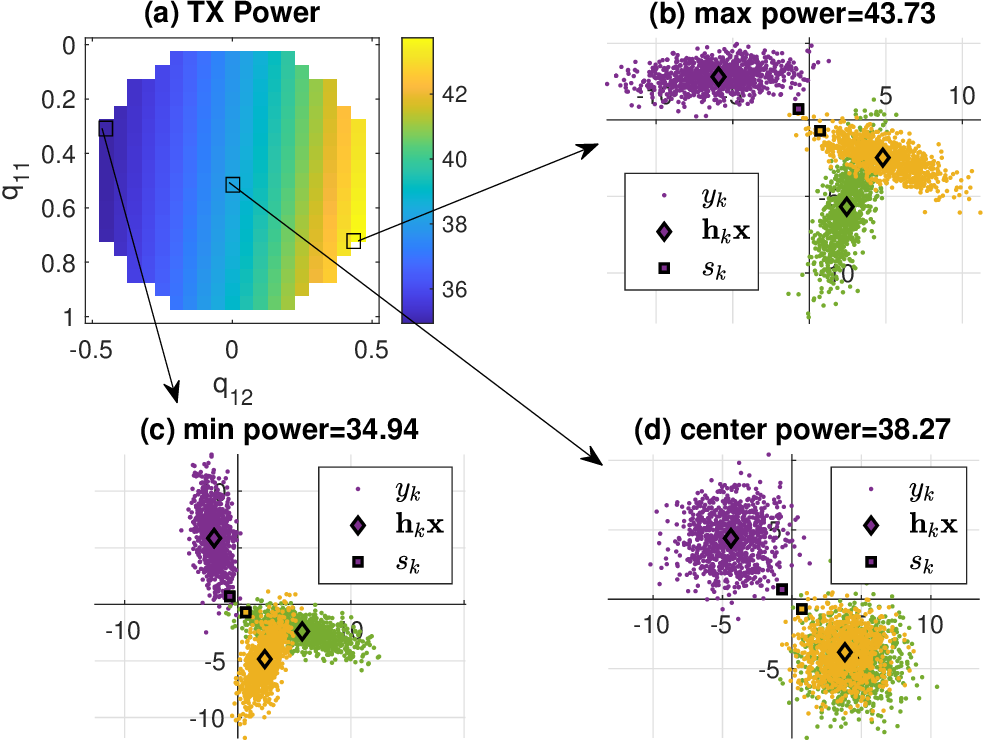}
\caption{NC-SLP based on minimizing power, $p=0.95$, $\rho^2=10$dB, $M=K=3$, QPSK. (a) transmit power for various $q_{11}$ and $q_{12}$; plots of the desired symbol ($s_k$), received noiseless symbols ($\mathbf{h}_k\mathbf{x}$), received symbols plus improper noise ($y_k$) when (b) TX power is the maximum; (c) TX power is minimum; (d) TX power is calculated with $\mathbf{Q}=\frac{1}{2}\mathbf{I}_2$.}
\label{robustSLPmap}
\end{figure}

\subsubsection{Robustness conditions on $\Q$}
Fig.~\ref{robustBLPmap} depicts the MSE as a function of $q_{11}$ and $q_{12}$ evaluated on a grid. Fig.~\ref{robustBLPmap}(a) shows the MSE averaged over $10^6$ symbols using Eq.~(\ref{eq:MSE}) for each corresponding $\mathbf{Q}$. As predicted by Lemma~1, the largest MSE consistently occurs at the center of the grid, where $\mathbf{Q}=\frac{1}{2}\mathbf{I}_2$. This result is consistent with our proof in Section~\ref{sec:robustBLP}.

On the other hand, Fig.~\ref{robustSLPmap}(a) plots the transmit power required to solve the SLP problem in~(\ref{opt:smartSLP}) for various $\mathbf{Q}$. The maximum transmit power appears on the boundary of the feasible region for $q_{11}$ and $q_{12}$, where $\mathbf{Q}$ is rank deficient. This implies that that the jamming disturbance ellipse around the noiseless symbol will degenerate to a straight line, pushing the desired symbol in only one direction. This result validates our proof of Lemma~\ref{lemma:robustSLP} in Section~\ref{sec:robustSLP}.

\subsubsection{Robustness with $\mathbf{Q}$ of different ranks}
\begin{figure}[!t]
\centering
\includegraphics[width=3in]{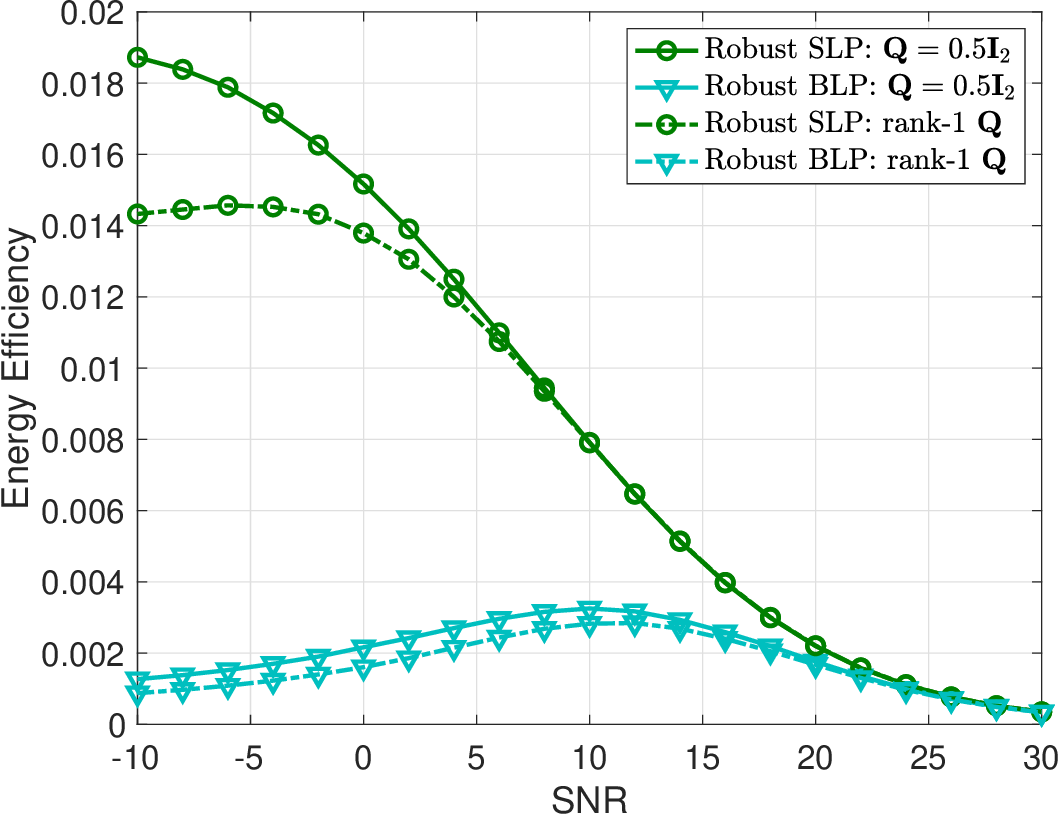}
\caption{EE for $\mathbf{Q}$ generated with different ranks, $\rho^2=10$dB, $M=K=8$, $p=95\%$, QPSK.}
\label{fig:rank12EE}
\end{figure}

From the above discussion, it is clear that the BLP and SLP approaches favor $\mathbf{Q}$ covariances of different ranks. To investigate the performance of the robust algorithms for different $\mathbf{Q}$, we plot EE in Fig.~\ref{fig:rank12EE} for $\mathbf{Q}=0.5\mathbf{I}_2$ and for cases averaged over various rank-1 $\mathbf{Q}$. Notably, robust SLP is more sensitive to the rank of $\mathbf{Q}$ than BLP, which supports the theoretical results of Section~\ref{sec:robust}.    


\section{conclusion}
This work has focused on enhancing block-level and symbol-level precoders in wireless communication systems in situations where the receivers are jammed by improper Gaussian interference. When the channel and statistics of the jammer signal are known, we derived pre-whitening transforms that decorrelated the interference for both BLP and SLP, and we also investigated how SLP could be modified to take IGI into account using only the transmit precoder design, without receiver preprocessing. Then we considered the case where the channel and statistics of the jammer are unknown, and developed robust BLP and SLP approaches for this case. In particular, we designed the MMSE BLP algorithm and the SLP approach such that they optimized their behavior for the worst-case IGI. Interestingly, we showed that the worst case IGI for MMSE BLP is in fact circular, while the worst case for SLP is a jammer with maximally non-circular covariance. Finally, we presented simulation results that demonstrated the advantage of taking IGI into account, and that validated theoretical observations concerning the worst-case jammer covariance scenarios for BLP and SLP.

\appendix
\subsection{Upper SM Constraint}\label{appendix:upperSM}
In this appendix, we show detailed steps to derive the upper SM constraint. Firstly, in order to decorrelate the real and imaginary parts of the received signals, we rotate them by $-\check{\alpha}_k$ so the IGI can be written in real-valued form as
\begin{align}
\check{\mathbf{R}}_k\bar{\mathbf{S}}_k^T\bar{\mathbf{c}}_k\nonumber=&\check{\mathbf{R}}_k\begin{bmatrix}y_k^R-\bar{\mathbf{H}}_{k}^1\bar{\mathbf{x}}\\y_k^I-\bar{\mathbf{H}}_{k}^2\bar{\mathbf{x}}\end{bmatrix}\nonumber\\
=&\begin{bmatrix}(y_k^R-\bar{\mathbf{H}}_{k}^1\bar{\mathbf{x}})\cos\check{\alpha}_k+(y_k^I-\bar{\mathbf{H}}_{k}^2\bar{\mathbf{x}})\sin\check{\alpha}_k\\
-(y_k^R-\bar{\mathbf{H}}_{k}^1\bar{\mathbf{x}})\sin\check{\alpha}_k+(y_k^I-\bar{\mathbf{H}}_{k}^2\bar{\mathbf{x}})\cos\check{\alpha}_k\end{bmatrix}\; , 
\end{align}
where $\check{\mathbf{R}}_k$ is the rotation matrix in Eq.~(\ref{rotationMatrix}) with angle $\check{\alpha}_k$, and $y_k^R$ ($y_k^I$) is the real (imaginary) part of user $k$'s received signal. As in Eq.~(\ref{newCovariance}), after rotation the covariance of the effective noise for user $k$ is
\begin{equation}
\check{\mathbf{G}}_k^{\star}=\check{\mathbf{R}}_k\check{\mathbf{G}}_k\check{\mathbf{R}}_k^H\label{newCovarianceSLP}\; ,
\end{equation}
which is diagonal.

The eigenvalues of $\check{\mathbf{G}}_k^{\star}$ are denoted by $\lambda_{1,k},\ \lambda_{2,k}$, and are equal to those of $\mathbf{G}_k$ due to the orthogonality of $\bar{\mathbf{S}}_k$ and $\check{\mathbf{R}}_k$. The confidence ellipse corresponding to the effective IGI can be expressed as
\begin{align}
&\frac{\left((y_k^R-\bar{\mathbf{H}}_{k}^1\bar{\mathbf{x}})\cos\check{\alpha}_k+(y_k^I-\bar{\mathbf{H}}_{k}^2\bar{\mathbf{x}})\sin\check{\alpha}_k\right)^2}{\lambda_{1,k}}\nonumber\\
&+\frac{\left(-(y_k^R-\bar{\mathbf{H}}_{k}^1\bar{\mathbf{x}})\sin\check{\alpha}_k+(y_k^I-\bar{\mathbf{H}}_{k}^2\bar{\mathbf{x}})\cos\check{\alpha}_k\right)^2}{\lambda_{2,k}}=\omega_k\; ,\label{EllipseSLP}
\end{align}
which can be rewritten as 
\begin{align}
f(y_k^R,y_k^I)&=\lambda_{2,k}(y_k^R\cos\check{\alpha}_k+y_k^I\sin\check{\alpha}_k+a_k)^2\nonumber\\
&\quad \quad+\lambda_{1,k}(-y_k^R\sin\check{\alpha}_k+y_k^I\cos\check{\alpha}_k+b_k)^2\nonumber\\
&=\omega_k\lambda_{1,k}\lambda_{2,k}\; ,
\end{align} 
where $a_k=\bar{\mathbf{H}}_{k}^1\bar{\mathbf{x}}\cos\check{\alpha}_k+\bar{\mathbf{H}}_{k}^2\bar{\mathbf{x}}\sin\check{\alpha}_k$ and $b_k=\bar{\mathbf{H}}_{k}^1\bar{\mathbf{x}}\sin\check{\alpha}_k-\bar{\mathbf{H}}_{k}^2\bar{\mathbf{x}}\cos\check{\alpha}_k$. Considering $y_k^I$ as a function of $y_k^R$, we set the following derivative to zero,
\begin{align}
\frac{d f(y_k^R,y_k^I)}{d y_k^R}&=2A_k\left(\cos\check{\alpha}_k+\sin\check{\alpha}_k\frac{d y_k^I}{d y_k^R}\right)\nonumber\\
&\quad\quad+2B_k\left(-\sin\check{\alpha}_k+\cos\check{\alpha}_k\frac{d y_k^I}{d y_k^R}\right)\nonumber\\
&=0\; ,
\end{align}
where $A_k=\lambda_{2,k}(y_k^R\cos\check{\alpha}_k+y_k^I\sin\check{\alpha}_k+a_k)$ and $B_k=\lambda_{1,k}(-y_k^R\sin\check{\alpha}_k+y_k^I\cos\check{\alpha}_k+b_k)$. It is easy to show that
\begin{equation}
\frac{d y_k^I}{d y_k^R}=-\frac{A_k\cos\check{\alpha}_k-B_k\sin\check{\alpha}_k}{A_k\sin\check{\alpha}_k+B_k\cos\check{\alpha}_k}\; ,
\end{equation}
which is equal to the slope of the tangent line to the ellipse at the point $(y_k^R,y_k^I)$. 

The two orange points on the ellipse, where the tangent line is parallel to the upper decision boundary with the slope $\tan\theta$, can be obtained by simultaneously solving 
\begin{empheq}[left=\empheqlbrace]{align}
&\tan\theta=-\frac{A_k\cos\check{\alpha}_k-B_k\sin\check{\alpha}_k}{A_k\sin\check{\alpha}_k+B_k\cos\check{\alpha}_k}\label{eq:positiveSlope}\; ,\\
&\frac{B_k^2}{{\lambda_{1,k}}}+\frac{A_k^2}{{\lambda_{2,k}}}= \omega_k\label{positiveABellipse}\; .
\end{empheq}
From Eq.~(\ref{eq:positiveSlope}) we can get $A_k=\kappa_k B_k$ where 
\begin{equation}
    \kappa_k=\frac{\sin\check{\alpha}_k-\cos\check{\alpha}_k\tan\theta}{\cos\check{\alpha}_k+\sin\check{\alpha}_k\tan\theta}=\tan(\check{\alpha}_k-\theta)\; .\label{eq:kappak}
\end{equation} 
Substituting this into Eq.~(\ref{positiveABellipse}) yields 
$
B_k=\pm\sqrt{\frac{\omega_k}{d_k}}
$
where 
\begin{equation}
d_k=\frac{1}{\lambda_{1,k}}+\frac{\kappa_k^2}{\lambda_{2,k}}\; .\label{eq:dk}
\end{equation}
Then $y_k^R,y_k^I$ can be found by solving the following equations
\begin{empheq}[left=\empheqlbrace]{align}
&-y_k^R\sin\check{\alpha}_k+y_k^I\cos\check{\alpha}_k=\pm e_k-b_k\; , \\
&f_ky_k^R+g_ky_k^I=a_k\lambda_{2,k}+b_kc_k\lambda_{1,k}\; ,
\end{empheq}
where 
$e_k=\sqrt{\frac{\lambda_{2,k}\omega_k}{\lambda_{1,k}d_k}}$, $f_k=\lambda_{2,k}\cos\check{\alpha}_k+\lambda_{1,k}\kappa_k\sin\check{\alpha}_k$, and $g_k=\lambda_{2,k}\sin\check{\alpha}_k-\lambda_{1,k}\kappa_k\cos\check{\alpha}_k$.
It is easy to show that
\begin{empheq}[left=\empheqlbrace]{align}
&f_k\cos\check{\alpha}_k+g_k\sin\check{\alpha}_k=\lambda_{2,k}\; ,\\
&f_k\sin\check{\alpha}_k-g_k\cos\check{\alpha}_k=\kappa_k\lambda_{1,k}\; .
\end{empheq}

The two orange points in Fig.~\ref{EllipseLine} are then given by
\begin{align}
&y_k^R=\bar{\mathbf{H}}_{k}^1\bar{\mathbf{x}}-u_k\;, \qquad y_k^I=\bar{\mathbf{H}}_{k}^2\bar{\mathbf{x}}+v_k\; ,\label{posPoint1}\\
&y_k^R=\bar{\mathbf{H}}_{k}^1\bar{\mathbf{x}}+u_k\;,  \qquad y_k^I=\bar{\mathbf{H}}_{k}^2\bar{\mathbf{x}}-v_k \; , \label{posPoint2}
\end{align}
where $u_k=\frac{g_ke_k}{\lambda_{2,k}}$ and $v_k=\frac{f_ke_k}{\lambda_{2,k}}$.
Defining the {\em upper SM} as $\delta_{u,k}^0$, the SM constraint corresponding to each of these points is given by
\begin{equation}
y_k^R\sin\theta-y_k^I\cos\theta\geq\delta_{u,k}^0\cos\theta\label{upperDB}\; .
\end{equation}

\subsection{Proof of Lemma \ref{lemma:robustBLP}}\label{appendix: robustBLP}
\begin{proof}
According to (\ref{LagrangianBeta}), the optimal precoding $\mathbf{P}_{opt}$ should satisfy
\begin{equation}
    \Tr\{\mathbf{H}_E\mathbf{P}\}-\beta^{-1}\Tr\{\mathbf{H}_E\mathbf{P}\mathbf{P}^T\mathbf{H}_E^T\}=-4K\beta^{-1}\; .\nonumber
\end{equation}
Then the mean square error (MSE) in (\ref{eq:lemma1MMSE}) can be rewritten as
\begin{align}
    &\mathbb{E}\{\|\beta^{-1}\mathbf{y}_E-\bar{\mathbf{s}}\|^2\}\label{eq:MSE}\\
    =&K-M+\frac{1}{2}\Tr\left\{\left(\mathbf{I}+\frac{1}{a}\sum_k \bar{\mathbf{H}}_k^T\mathbf{G}_k^{-1}\bar{\mathbf{H}}_k\right)^{-1}\right\}\; .\nonumber
\end{align} 
Define 
\begin{eqnarray}
  \Dbf & = & \Ibf_{2M}+\frac{1}{a} \sum_k \bar{\Hbf}_k^T \Gbf_k^{-1} \bar{\Hbf}_k \\
    & = & \Ibf_{2M}+\frac{1}{a} \Fbf^T \Gbf^{-1} \Fbf \; ,
\end{eqnarray}
where 
\begin{eqnarray}
    \Fbf^T & = & \left[ \bar{\Hbf}_1^T \; \cdots \; \bar{\Hbf}_K^T \right]\; , \\
    \Gbf & = & \text{blockdiag} \, \{\Gbf_1, \cdots , \Gbf_K\} \; ,
\end{eqnarray}
so that maximizing the MSE is equivalent to maximizing $\Tr\{\mathbf{D}^{-1}\}$. Note that since we assume the jammer and noise power are fixed, then $\sigma_k^2=\Tr\{\G_k\}$ is fixed.

Since $\Gbf_k=\Gbf_k^T$, we can write $\Gbf_k=\Ubf_k\Lambdabf_k\Ubf_k^T$, where $\Ubf_k\Ubf_k^T=\I_2$ and
\begin{equation}
    \Lambdabf_k = \left[ \begin{array}{cc} \lambda_{k} & 0 \\ 0 & \sigma^2_k - \lambda_{k} \end{array} \right].
\end{equation}
Next, take the gradient of $\Tr\{\D^{-1}\}$ with respect to $\lambda_k$:
\begin{eqnarray*}
\frac{d}{d \lambda_k} \Tr\{\D^{-1}\} & = & - \Tr\left\{ \D^{-1} \left( \frac{d}{d \lambda_k} \D\right) \D^{-1} \right\} \\
& = & \frac{1}{a} \Tr \left\{ \A \left( \frac{d}{d\lambda_k} \G_k \right) \A^T \right\} \\
& = & \frac{1}{a} \Tr \left\{ \U_k \left[\begin{array}{cc} 1 & 0 \\ 0 & -1 \end{array}\right]\U_k^T \A^T\A \right\} \; ,
\end{eqnarray*}
where $\A=\D^{-1}\bar{\H}_k^T\G_k^{-1}$. We can write
\begin{equation}
    \U_k \left[\begin{array}{cc} 1 & 0 \\ 0 & -1 \end{array}\right]\U_k^T = \tilde{\U}_k\left[\begin{array}{cc} -1 & 0 \\ 0 & 1 \end{array}\right] \tilde{\U}_k^T \; ,
\end{equation}
where $\U_k$ is unitary, symmetric and improper, while $\tilde{\U}_k$ is unitary, non-symmetric but still proper. Thus, 
\begin{equation}
\frac{d}{d \lambda_k} \Tr\{\D^{-1}\} = \frac{1}{a} \Tr \left\{ \left[\begin{array}{cc} 1 & 0 \\ 0 & -1 \end{array}\right]\U_k^T \A^T\A\U_k \right\} \; .
\end{equation}
The $2\times 2$ matrix $\Gammabf_k \triangleq \U_k^T \A^T\A\U_k$ will be a proper matrix when $\mathbf{G}_k$ is proper, which means it can be written as 
\begin{equation}
    \Gammabf_{k} = \left[ \begin{array}{cc} \mathcal{R}(\gamma_{k}) & -\mathcal{I}(\gamma_{k}) \\ \mathcal{I}(\gamma_{k}) & \mathcal{R}(\gamma_{k}) \end{array} \right] \; ,
\end{equation}
for some complex $\gamma_k$. Thus, it is clear that $\frac{d}{d\lambda} \Tr\{\D^{-1}\} = 0$ when $\mathbf{G}_k=\frac{\sigma^2_k}{2}\I_2$, which is thus a stationary point for $\Tr\{\D^{-1}\}$ where it is maximized. Since the above is true for all $k$, we can conclude that the worst case for the MMSE block precoder occurs when the jammer covariance matrix is circular. \qedsymbol
\end{proof}

\bibliographystyle{IEEEtran}
\bibliography{IEEEabrv,Ref}
\end{document}